\documentclass[11pt]{article}
\usepackage[letterpaper]{geometry}
\usepackage{amsmath,amsthm,amsfonts,amssymb}
\usepackage{enumerate,color,xcolor}
\usepackage{graphicx}
\usepackage{url}

\usepackage{hyperref}
\usepackage{dsfont}
\usepackage{subcaption}

\numberwithin{equation}{section}
\theoremstyle{plain}
\newtheorem{theorem}{Theorem}[section]

\newtheorem{lemma}[theorem]{Lemma}

\theoremstyle{definition}

\theoremstyle{remark}

\usepackage{comment}

\newcommand{\rev}[1]{\textcolor{black}{#1}}
\newcommand{\gao}[1]{\textcolor{black}{#1}}

\graphicspath{  {images/}  }

\allowdisplaybreaks

\usepackage{natbib}

\begin{document}
	\begin{center}
		\Large \bf  Optimal Market Making in the Presence of Latency
	\end{center}
	\author{}
	\begin{center}
		{Xuefeng
			Gao}\,\footnote{Department of Systems
			Engineering and Engineering Management, The Chinese University of Hong Kong, Shatin, N.T. Hong Kong;
			Email: xfgao@se.cuhk.edu.hk.},
		Yunhan Wang\,\footnote{
			Email: yunhanwang25@163.com.
		}
	\end{center}
	
	\begin{abstract}
This paper studies optimal market making for large--tick assets in the presence of latency.
We consider a random walk model for the asset price, and formulate the market maker's optimization problem using Markov Decision Processes (MDP). We characterize the value of an order and show that it plays the role of one-period reward in the MDP model.
Based on this characterization, we provide explicit criteria for assessing the profitability of market making when there is latency. Under our model, we show that
a market maker can earn a positive expected profit if there are sufficient uninformed market orders hitting the market maker's limit orders compared with the rate of price jumps, and the trading horizon is sufficiently long. In addition,
our theoretical and numerical results suggest that
latency can be an additional
source of risk and latency impacts negatively the performance of market makers. 
\end{abstract}

\section{Introduction}

A market maker in a security market provides liquidity to other investors by quoting bids and offers, hoping to make a profit from the bid-ask spread while avoiding accumulating a large net position in the assets traded. Market makers play a crucial role in financial markets,
\rev{because} the liquidity they offer allows investors to obtain immediate executions of their orders, and this flexibility facilitates market efficiency and functioning.
Traditionally in equity markets, there are `official' or designated market makers who have entered into contractual agreements with exchanges and they are under certain affirmative obligations to stand ready to supply liquidity.
In recent decades, major financial markets have became electronic, and a modern exchange is typically operated as an electronic limit order book system under a price-time priority rule, where all the outstanding limit orders are aggregated for market participants to view, see \cite{Parlour2008} for an overview. As a result, any professional trader can adopt market making as a trading strategy, often through computer-based electronic trading decisions and automated trade executions. This work focuses on such `unofficial' electronic market makers, who can enter and exit market at will or set buy
and sell prices at any level.

In this paper, we study optimal market making with latency taken into consideration. There are many different definitions for the term “latency” in the literature. Following \rev{\cite{Hasbrouck2013, Cartea2018}},
we define latency as `the time delay between an exchange sending market data to a trader, the trader processing information and deciding to trade, and the exchange receiving the order from the trader'. Latency has always been important for participants including market makers in electronic trading where markets trade at high speed, see e.g. \cite{AFM, gomber2015high, Moallemi}. Indeed, it is essential for market makers to be able to respond
rapidly to newly available price information and changing market conditions by placing or
canceling orders at low latency.

We consider the market making problem of an agent trading a large-tick asset. Following \cite{Eisler2012, Rosenbaum2015}, a
large-tick asset is defined as an asset with the bid-ask spread almost always equal to its minimum value of one tick. Large-tick assets are typically liquidly traded. Indeed, if the asset price is low, then the tick size relative to the asset price is larger, and hence liquidity providers can earn a higher revenue margin (\cite{YaoYe2018}).
Empirically, it has also been found that electronic market makers take on a prominent role in liquidity provision for large-tick stocks, see e.g. \cite{Ohara2015}.

We formulate a discrete-time trading model to investigate optimal market making in the presence of latency.
We consider an `unofficial' market maker trading on a single venue and we model the asset price using a Compound Poisson process. Similar as in \cite{Moallemi, Stoikov2016}, we assume a constant (absolute) latency $\Delta \tau \ge 0$.
The maker can quote a bid order for one unit and an ask order for one unit at \textit{any discrete} prices periodically with a deterministic period
length $\Delta t > \Delta \tau$, with the goal of maximizing his expected profit within a finite horizon.
At each period, the market maker can choose to cancel the old orders that remain in the orderbook, replace them with new orders at appropriate prices, or do nothing.
At the end of the horizon, the market maker liquidates his position and goes home flat.
The optimization problem faced by the market maker is formalized as a discrete-time finite-horizon Markov Decision Process (MDP).
As in the existing optimal market making studies (see, e.g., \cite{Avellaneda, Cartera-book}), we do not include rebate and fees for orders in the model for simplicity. In addition, the market maker we consider is small in the sense that he does not impose price impact onto the market and the decisions do not involve the size of orders.

To our knowledge, our model is the first to study optimal market making with latency.
By a delicate analysis of the MDP together with numerical experiments, we make the following contributions.

First, we characterize and highlight the role of order value in optimal market making. The value of an order measures the difference of its execution price with the fair value of the asset at the time of order execution.
It is referred to as the expected profit of an order in the market microstructure literature, and has been widely used in equilibrium models in finance and economics, see e.g., \cite{sandas2001, Hoffmann2014}. We show that in our MDP model for optimal market making, the value of bid and ask orders from the market maker in each period essentially plays the role of one-period reward, where there is a terminal inventory penalty arising from liquidating inventory positions. We give an explicit characterization of the value of an order under our model and provide simple condition to decide when the order value is positive in the possible presence of latency. See Theorems~\ref{value of orders tau} and \ref{BellmanThm} .



Second, we provide explicit criteria on when electronic market making is profitable under our model. Based on the order value and the structure of the value function for the MDP, we find that the profitability condition reduces to comparing two rates: (1) the rate of the market maker's limit order being `adversely' filled when the market trades through its limit price, e.g., due to a surge of flow of market order eats through the order book and induces price jumps; and (2) the rate of small (or `uninformed') market order flows that hit the market maker's limit orders at best quotes without moving the price.
We show that if the latter rate is smaller than or equal to the former, then the order values are non-positive, and hence
the market maker can not make any profit in trading regardless of the number of quoting periods. On the other hand, if the latter rate is larger, then there exists a bid-ask quote pair whose order values are positive. This allows us to construct an explicit quoting policy so that when the number of quoting periods is sufficiently large, the market maker can earn positive \rev{expected} profit. How many periods is needed depends on model parameters including latency, and we can upper bound it explicitly using values and fill probabilities of bid and ask orders. See Theorem~\ref{net profit}.

Third, we provide qualitative insight on the importance of latency in optimal market making.
Our theoretical and numerical results show that
latency can be an additional
source of risk for market makers.
On one hand, the market maker can send new orders in each period.
Due to the possible price motion in the latency time
window, such new orders will enter undesirable price positions and bear the risk of being executed at an improper price or not being executed at all. This can be reflected as in the decreased value of such new orders. On the other hand, the market maker may have outstanding old orders in the orderbook that have not been executed yet in each period.
Under unfavorable market conditions, the value of these outstanding orders is negative (i.e., earns negative expected profit), and the loss the market maker suffers increases with latency. See Sections~\ref{sec theoretical} and \ref{sec: numerical results}.

\textbf{Related Literature.} We explain the differences between our work and the existing studies.

Our paper is related to the growing body of literature on optimal market making in the quantitative finance literature, see, e.g., \cite{Avellaneda, Gueant2013, Guilbaud, Cartea, Cartera-book, Fordra2015, AitSahhalia} and the references therein. These papers mainly use continuous-time stochastic control approaches to determine optimal quoting strategies for market makers in an expected utility framework. Among them, \cite{Fordra2015} and Chapter~10.2.2 in \cite{Cartera-book} considered large-tick liquid stocks. Both studies focused on a single market maker who submits limit orders only at best bid and best ask prices, so the control is a pair of predictable processes valued in $\{0,1\}$, indicating whether the market maker posts an order or not at the best quotes. \cite{Cartera-book} modeled the asset price as a Brownian motion, and the executions of orders are independent from the price. \cite{Fordra2015} considered a rich model where asset price is modeled by a Markov renewal process, and the market order flows hitting the market maker's quotes is modeled by a Cox process subordinated to asset price. Our model is similar (albeit simpler) in spirit to the one in \cite{Fordra2015}. Our work complements these studies and differs from them in two main aspects: first,
we explicitly consider latency in optimal market making; second,
to focus on the latency effect, we consider simpler price dynamics using a Compound Poisson process, but we allow the market maker to submit orders at any discrete prices, including using market orders.

Our work is also related to the study of latency in algorithm and high frequency trading.
\cite{Moallemi} quantified the cost of latency in an optimal execution problem of a
representative investor using the volatility and the bid-ask spread of an asset. \cite{Stoikov2016} studied a pure market order strategy for optimal liquidation by taking latency into account.
\cite{Lehalle2017} studied an optimal control problem of letting a single limit order rest at the best bid price or cancelling it in the presence of latency.
\gao{\cite{Cartea2018} showed the effect of latency for liquidity-taking orders and studied how to improve fill ratios of marketable orders in foreign exchange markets.
\cite{Cartea2019} analyzed how a liquidity taker chooses the price limit of marketable orders and developed an optimal strategy that balances the tradeoff between the number of missing trades and the costs of walking the book
when there is latency.
}
 \cite{Hoffmann2014} found that low latency allows liquidity providers to reduce their adverse
selection costs by revising stale quotes before picked off.
See also \cite{Hasbrouck2013, Kirilenko, Menkveld2016, Baron2018} and the references therein for further related studies.
Our work differs from these studies primarily in that the problem we consider is optimal market making in the presence of latency.

Finally, our paper is related to studies on
the profitability of market making, see, e.g. \cite{Charkraborty2011} and Chapter 17 of \cite{Bouchaud}. \cite{Charkraborty2011} related the profitability with time series models for the price dynamics and they showed that market making is generally profitable if the price series exhibit mean reversion. \cite{Bouchaud} analyzed whether a market maker who always quotes at best bid and best ask prices can be profitable, using a model-free analysis. They found that such simple market making strategies yield negative profits, and the main challenge for large-tick stocks is to gain time priority in the queue. Our work is different from these studies in that we study the profitability of optimal market making strategies using a stylized MDP model with price modeled as a Compound Poisson process. In particular, the \rev{expected} profit of the market maker in our model is lower bounded by zero as not sending any quotes is also a feasible strategy.

\textbf{Organization.} The rest of the paper is organized as follows. In Section~\ref{model}, we describe the model and formulate the market maker's optimization problem. In Section~\ref{sec:value-order} we define the value of an order. In Section~\ref{sec theoretical}, we present the main theoretical results. In Section~\ref{sec: numerical results}, we present numerical results for a representative large-tick stock using orderbook data from NASDAQ. Finally, Section~\ref{sec:conclusion} concludes. Some technical details of the model and all the proofs of the theoretical results are given in the Appendix.

\section{The MDP model for optimal market making}\label{model}
In this section, we present a finite-horizon MDP to model an electronic market maker who aims to maximize his expected terminal wealth in the presence of latency.


\subsection{Market making in the presence of latency}

To discuss the process of market making in the presence of latency,
we first describe market primitives.
We assume that the bid-ask spread of the asset is exogenously given at constant one tick, which is typical for large-tick liquid assets \citep{Rosenbaum2015}.
We also assume that the \rev{mid-price} $\{p(t): t \ge 0\}$ is a compound Poisson process: 
\begin{equation}
	p(t)=p(0)+\sum\limits_{i=1}^{\mathcal{N}(t)}X_i,
\end{equation}
where $\{\mathcal{N}(t): t\geq0\}$
is a Poisson process with a rate $\lambda$ and $(X_i)_{i=1,2,...}$ are independent and identically distributed random variables taking values $+1$ (tick) and $-1$ both with probability $0.5$. See, e.g., \cite{AitSahhalia, Guo} for similar models. We use symmetric jump size $X_i$ as market makers typically have no directional opinion on the assets they trade.

Next, we describe the market maker's periodic quoting process in the presence of latency. See Figure~\ref{fig:mm-illustration} for a graphical illustration.
During a finite horizon $[0, T]$, the market maker takes actions $N+1$ times
every $\Delta t$ time units. The market maker experiences a constant absolute latency $\Delta\tau \in [0, \Delta t)$.\footnote{Technically, if $\Delta\tau>\Delta t$, say, $\Delta\tau= k \cdot \Delta t$ for an integer $k \ge 2$, then one needs to consider \rev{MDP} with delays \cite{Katsikoulos2003}, which is outside of the scope of this paper.} This latency captures the time delay between an exchange sending message to the maker, the maker processing information and deciding to trade, and the exchange receiving the actions or orders from the maker.

Specifically, starting from time zero, the exchange sends messages continuously to the market maker. The messages contain the asset price, the maker's outstanding orders (if any), as well as other market information such as the state of the orderbook. Outstanding orders refer to limit orders (sent by the maker) waiting to be filled in the limit order book. {In each period, except for the last one, the market maker can take three kinds of actions for both ask and bide sides. Taking the ask side as an example, the first one is that the maker sends an ask cancellation instruction and a new ask order. The size of the order is fixed at one unit (e.g., 100 shares). The cancellation instruction will automatically cancel the maker's any outstanding ask orders (if any). The second one is that the maker sends an ask cancellation instruction and does not send any ask orders. The third one is that the maker does not send any ask orders or order cancellation instructions.  
It is similar for the bid side. With these actions, it is readily seen that in our model the market maker has at most one ask/bid outstanding order in the order book at any time.}
For the last period, the maker unwinds all his inventory using a market order.

We define the $N+1$ action times (including the last period) as the time when the message is sent by the exchange on which the corresponding actions are made, i.e.,
	\begin{equation} \label{eq:t-i}		
	t_i:=i\cdot\Delta t,\;  \quad\text{for} \quad i=0,1,2,...,N.
	\end{equation}
Due to latency, the time when the orders or cancellation instructions of $i$-th action (if the action is not doing nothing) enter into the limit order book is \rev{$t_{i+}:=t_i+\Delta \tau$}. In particular, the
time when the unwinding market order arrives at the exchange is
	\rev{\begin{equation*}
	t_{N+}:=N \cdot \Delta t+\Delta \tau.
	\end{equation*}}
For simplicity, throughout the paper, we say an order/instruction is sent at time $t$ to mean that the order/instruction sent is based on the information at time $t$.
 We assume that the maker quotes as many times as possible. As \rev{$t_{N+}\le T$}, we have $N=\lfloor \frac{T-\Delta\tau}{\Delta t}\rfloor$, where $\lfloor x \rfloor$ is the greatest integer that is less than or equal to $x \in \mathbb{R}$.
	 \begin{figure}
	 	\centering
	 	\includegraphics[width=1\textwidth]{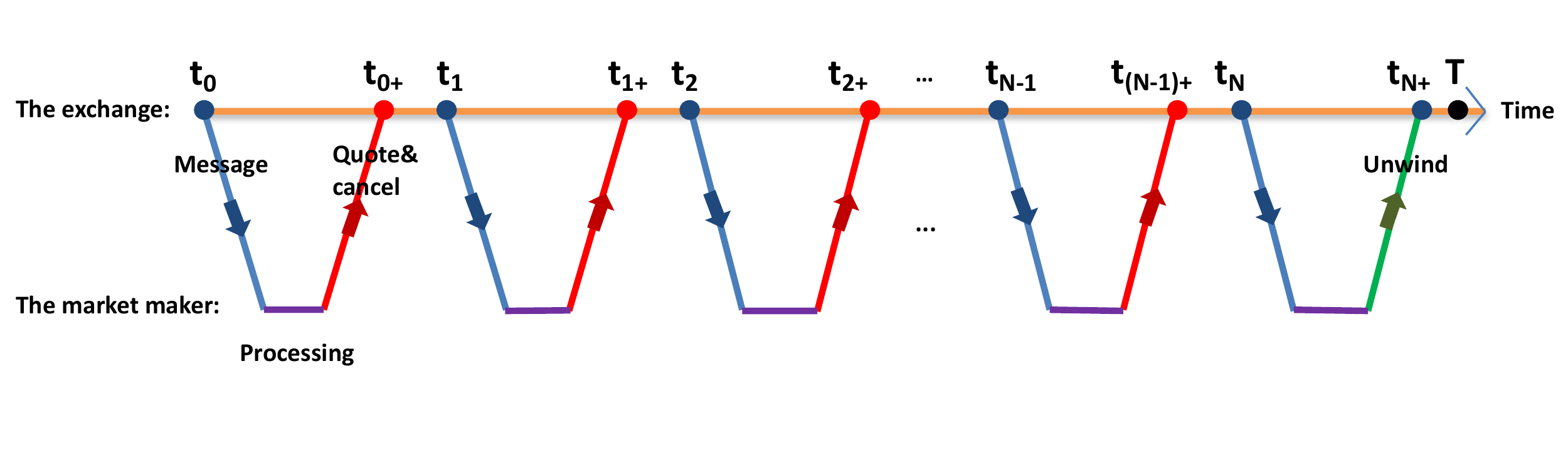}
	 	\caption { \label{fig:mm-illustration} An illustration of the trader's market making process. Here, $t_{i+}=t_i+\Delta \tau$, where $\Delta \tau$ is the latency.}
	 \end{figure}

As standard in the literature (see e.g., \cite{Gueant2013, Cartera-book}), to control the inventory risk, we assume the market maker's inventory (number of units of an asset held) is constrained by a lower bound $ \underline{q}$ and an upper bound $\overline{q}$. Here $\underline{q}$ and $\overline{q}$ are two fixed integers with $\underline{q} < -1 $ and $\overline{q}>1$.

\subsection{Order executions}\label{orderexecution}
We now describe the executions of the market maker's \rev{limit} orders in our model.
For illustrations, we consider an ask order attached with an arbitrary limit price. Considering a bid order is similar. Note that for large-tick assets which typically have long queues of limit orders at best quotes, most market orders do not walk the orderbook and limit orders are typically executed at best bid and best ask prices.
In our setting, when the market maker's ask order arrives at the exchange, it can be executed in one of the following ways:
	\begin{enumerate}
		\item If the limit price of an ask order is smaller or equal to the market best bid price, the order will get executed immediately\footnote{\rev{In some venues, market makers may not be allowed to trade with market makers, and because of this, the spreads can be inverted for short periods of time (allowing arbitrage). We do not consider such scenarios.}}. 
\vspace{2mm}
		\item Otherwise, the ask order will enter into the order book and it will get executed at its limit price if
		\begin{enumerate}
			\item either when the order sits at the best ask price and the best bid price jumps up by one tick, crossing the limit price of the ask order; \vspace{1mm}
			\item or the total time the order spends at the best ask price exceeds an exponential random variable with rate \rev{$\lambda^+$}.

		\end{enumerate}
	\end{enumerate}

Case 1 can occur, for example, when the maker sends a limit ask order and the mid-price of the asset moves up during the latency period. The execution price of the ask order is the market best bid price at the moment the ask order reaches the exchange.


In Case 2(a), the ask order is filled when the market trades through its limit price, i.e., it is among the last in the queue to get filled just as the market moves higher. In our model (asymmetric information is not explicitly modeled), this serves the role of `adverse selection' for market makers.
Case 2(a) can happen, e.g., if there is a surge of buy orders into the best ask.

In Case 2(b),
the mid-price does not move when the ask order is filled, i.e., the order leans against the queue behind it.
We interpret \rev{$\lambda^+$} as the hit rate of small or `uninformed' buy market order flow that matches with this market maker's ask order at the best ask.

 Similarly, we use \rev{$\lambda^-$} for the hit rate of small or `uninformed' sell market order flow that matches with this market maker's buy order at the best bid price. These `uninformed' buy and sell market order arrivals are assumed to be mutually independent and independent of the price process.

We remark that the rates \rev{$\lambda^{\pm}$} may depend on the rate of price change $\lambda$ and the latency $\Delta \tau$ the market maker experiences. This would not affect our mathematical results in Section~\ref{sec theoretical}. We suppress such dependence for notational simplicity.

\subsection{State space and admissible action space}

We now describe the system state and the admissible action space for the MDP model.
We write the state of available information at $t$ as $s(t)$, and the set of extended integers $\mathbb{\overline{Z}}:=\mathbb{Z}\cup{\{\pm\infty\}}.$ At any time $t\in[0,T]$, the market maker uses \rev{$s(t):=(w(t),p(t),q(t),r^+(t),r^-(t)) \in \mathbb{Z}^3 \times \mathbb{\overline{Z}}^2$} for making decisions, where $w(t)$ is the market maker's wealth, $p(t)$ is the \rev{mid-price} of the asset, $q(t)$ is the maker's inventory. In addition, \rev{$(r^+(t),r^-(t))$} represents the ask-bid outstanding quote pair. More precisely,
\rev{$p(t)+0.5+r^+(t)$} is the price of the marker's ask outstanding order at time $t$ with \rev{$r^+(t)=\infty$} meaning there is no ask outstanding order at time $t$.
\rev{Note that $p(t)+0.5$ is the best ask price of the asset at time $t$ and hence $r^+(t)$ is the relative price of the ask outstanding order comparing with the best ask price.}
Similarly, \rev{$p(t)-0.5-r^-(t)$} is the price of the marker's outstanding bid order at time $t$ where \rev{$r^-(t)=\infty$} means there is no bid order at time $t$.
 For the tractability of the MDP model, we do not include the order book status, price signals, and other market information in $s(t)$.

From the Poisson assumptions, it is easy to see that the sample paths of $s(t)$ are right-continuous with left limits. For $i=0,1,2,...,N$, denote
\rev{\begin{equation}
s_i=(w_i,p_i,q_i,r^+_i,r^-_i):=s(t_i-), \quad \text{and} \quad
s_{i+}=(w_{i+},p_{i+},q_{i+},r^+_{i+},r^-_{i+}):=s(t_{i+}-).
\end{equation}}
The $N+1$ states $s_i, i=0,1,...,N$ relate to the maker's $N+1$ actions (quote, cancellation or unwinding) and \rev{$s_{N+}$} is the final state that is just before the time to end market making. These $N+2$ states $s_0, s_1, ..., s_N, \rev{s_{N+}}$ are the system states for the discrete-time MDP and $\rev{s_{i+}}, i=1,2,...,N-1$, are intermediates to compute dynamics of these system states.\footnote{We use the left limits of the underlying continuous-time state process for the discrete-time state, which is a convention in the continuous-time stochastic control literature.} We also call the $N+1$ times $t_i$, $i=0,1,..,N$, the decision epochs in the discrete-time MDP (though at $t_N$, the market maker simply liquidates the position using a market order).

We can now write down the state space.
Clearly we have \rev{$r^\pm(t)\ge 0$}. Note that when the maker's inventory reaches the lower/upper bound, he should not have any ask/bid outstanding orders at each decision epoch, since otherwise the inventory may exceed the two bounds due to possible execution of those outstanding orders. Hence, the state space $S$  is given as follows.
	\rev{\begin{equation*}\label{S}
	\begin{split}
	S := \{(w, p, q, r^+, r^-) : \quad &(w, p, q) \in\mathbb{Z}^3,  \; (r^+, r^-) \in \mathbb{\overline{Z}}\times\mathbb{\overline{Z}}, \;  \underline{q}\leq q\leq\overline{q},  \; \rev{r^\pm\ge 0}, \\
	&	\text{if} \; q =\underline{q},  \; \text{then} \;r^+= \infty,  \; \text{and if}\;q = \overline{q},   \; \text{then} \;  \rev{r^- = \infty}\}.\\
	\end{split}
	\end{equation*}}

Next, we describe the admissible action space. The maker is allowed to send market orders, limit orders or not send any orders. For each period, when the market maker receives the system state $s=(w,p,q,\rev{r^+,r^-})$, {we use a pair \rev{$(\delta^+,\delta^-)$} representing the action of the maker, where $\rev{\delta^+,\delta^-}\in\mathbb{\overline{Z}}\cup\{o\}$. \rev{$\delta^+\in\mathbb{Z}$} means the maker quotes an ask order at price \rev{$p+0.5+\delta^+$}, i.e., \rev{$\delta^+$} is the relative price comparing with \rev{the best ask price},  together with an instruction to cancel his previous ask outstanding order.
\rev{$\delta^+=\infty$} means the maker cancels his ask outstanding order but does not send any new ask orders. \rev{$\delta^+=-\infty$} means the maker sends a sell market order with an ask cancellation instruction. \rev{$\delta^+=o$} means the maker does nothing for the ask side. The difference between \rev{$\delta^+=\infty$ and $\delta^+=o$} is whether or not the maker sends a cancellation instruction for the ask outstanding order.} It is similar for the bid side. Specifically, \rev{$\delta^-\in\mathbb{Z}$ means the maker quotes an bid order at price \rev{$p-0.5-\delta^-$}}. \rev{$\delta^-=\infty$} means no buy order but a bid cancellation instruction is sent and \rev{$\delta^-=-\infty$} means a buy market order with a bid cancellation. For convenience, {we also call \rev{$\delta^\pm=\pm\infty$} relative prices and use the notation $\mathbb{\overline{Z}}^o:=\mathbb{\overline{Z}}\cup\{o\}$ for the set of possible decisions.}

We use $A_s$ to denote the set of admissible actions \rev{$(\delta^+, \delta^-)$} for state $s$, such that the inventory of the market maker always stays in the interval $[\underline{q}, \overline{q}]$ at any time. For the mathematical expression of $A_s$, see Section~\ref{sec:As}. To set up the MDP, we also need to describe the dynamics for the system state. We refer the readers to Section~\ref{sec:dynamics} for details.

\subsection{Optimization problem for the market maker}
In this section, we formulate the optimization problem for the market maker. The maker quotes bid and ask orders at each period, and aims to maximize his expected terminal wealth after he unwinds the position at the end of the trading horizon. Costs of trading such as IT and compliance are not considered.

We first give the expression for the market maker's terminal wealth, denoted as \rev{$\mathcal{W}$}. Suppose just before \rev{$t_{N+}$}, the state \rev{$s_{N+}=s(t_{N+}-)$ is $(w_{N+},p_{N+},q_{N+},\rev{r^+_{N+},r^-_{N+}})$}. Then it is easy to see that
\rev{\begin{equation}\label{TW}
\begin{split}
\mathcal{W}&:= w_{N+}+p_{N+}q_{N+}-0.5\left|q_{N+}\right|.
\end{split}
\end{equation}}
That is, if the market maker has positive inventory \rev{$q_{N+}>0$,} then the maker unwinds the position by sending a market sell order and the execution price is the best bid price \rev{$p_{N+}-0.5$}. Similarly, if \rev{$q_{N+}<0$}, then the market maker sends a market buy order which will be filled at the best ask price \rev{$p_{N+} + 0.5$}. We do not consider the price impact of such a clean-up trade. This is reasonable as long as the market maker's inventory bounds do not exceed the market depth of best quotes in the order book. For large-tick assets, it is typical that there are large volumes of limit orders sitting at best quotes.

Now we can formulate the optimization problem of the market maker as follows:
\rev{\begin{equation}\label{v_0}
	\begin{split}
		v_0(s)=v_0(w,p,q,r^+,r^-) :=\sup\limits_{\pi} E^{\pi}[\mathcal{W}\mid s_0=(w,p,q,r^+,r^-)],\\
	\end{split}
\end{equation}}
where the supremum is taken over all Markovian admissible policies. Specifically, we have each Markov policy $\pi=(f_0,f_1,...,f_N)$, where $f_i(\cdot)\; \text{is a mapping from} \;S \;\text{to} \;\mathbb{\overline{Z}}^o \times\mathbb{\overline{Z}}^o  \;  \text{such that for all}  \; s\in S, \; f_i(s) \in A_s$, the admissible action space.
This function $v_0$ is called the value function starting from the $0$-th period.  We can also define the value function starting from $i$-th period, $i=1,2,...,N,\rev{N+}$, as follows:
\rev{\begin{equation}\label{v_i}
	\begin{split}
		&v_i(w,p,q,r^+,r^-) := \sup\limits_{\pi} E^{\pi}[\mathcal{W}\mid s_i=(w,p,q,r^+,r^-)],\; i=1,2,...,N,\\
		&v_{N+}(w,p,q,r^+,r^-):= w+pq-0.5\left|q\right|.
	\end{split}
\end{equation}}

Mathematically, in Equations (\ref{v_0}) and (\ref{v_i}), the existence of expectations is not trivial since the \rev{$\mathcal{W}$} is not bounded. To address this issue, there is a standard method using an integrable bounding function to bound value functions. For our MDP, one can use the bounding function $C(\left|w\right|+\left|p\right|+1)$ which can be verified to be integrable, where $C$ is a constant that is independent of the state $s$. For simplicity, we omit the proof and refer readers to \cite{Bauerle2011} for this method.

\subsection{The Bellman equation}
\noindent
As we have formulated the market making problem as MDP, standard arguments yield the following Bellman equation for the value functions:
\rev{\begin{equation}\label{Bellman}
v_i(s)=\left\{
\begin{aligned}
&w+pq-0.5\left|q\right|, &\quad i=N+,\\
&E[v_{N+}(s_{N+})\mid s_N=s], &\quad i=N,\\
&\sup\limits_{(\delta^+,\delta^-)\in{A_s}} E^{(\delta^+,\delta^-)}[v_{i+1}(s_{{N+1}})\mid s_i=s],&\quad i=0,1,...N-1,
\end{aligned}
\right.
\end{equation}}
where the superscript \rev{$(\delta^+,\delta^-)$} in the expectation notation means that in the $i$-th action, the ask action is \rev{$\delta^+$}, and the bid action is \rev{$\delta^-$}.
One can readily prove using the theory of upper semi-continuous MDP that the supremum operators in Equation \eqref{Bellman} can be attained, which generates an optimal policy for the MDP. As the argument is standard (see, e.g., \cite{Bauerle2011}), we omit the proof.



\subsection{Model limitations}
Before we present the main results, we briefly
discuss the main limitations of our model.
First, market makers in practice make decisions based on microstructure information such as the state of the order book and the queue positions of their orders when trading large-tick assets. We work with a stylized MDP model to include latency in market making, and
do not include these features for tractability purposes. Through our simplified model, we can find that the key ingredient in market-making is the value (expected profit) of an order. In practice, a market maker can use the microstructure signals and latency information to predict the order value and combine with our results to decide whether it is profitable to trade a particular large--tick asset. In addition, in our model,
the market maker is risk neutral in the sense that he maximizes the expected terminal wealth. If the market maker's expected profit is risk-adjusted, then the profitability condition we provide in this paper becomes a necessary condition for such a market maker to earn a positive risk-adjusted profit. Finally, we do not directly model the competition of market makers. There can be other marker makers with similar strategies and information in addition to the agent we consider, and lower latency can allow one to obtain a better position in order book queues.
Our model captures the relative speed advantage of the agent in a parsimonious manner by directly specifying the rate at which the `uninformed' market order flows hit the agent's limit orders, and this rate may decrease with the latency that the agent experiences (with other things fixed).


\section{Value of an order} \label{sec:value-order}
In this section, we introduce and define the value of an order (also called `order value')
that plays a critical role in our analysis of the optimal market making problem.

The value (or the expected profit) of an order measures the difference of its execution price with the `fundamental value' of the asset at the time of order execution. For example, if 
one uses the asset mid-price at the time of order execution as the fundamental value, then the value of an one-unit market order is $-0.5$ ticks for a large-tick asset. That is, a trader pays half of the bid-ask spread using a market order.

To define the value of an order when there is latency, we note that the market price might have moved between the moment an order sent by the market maker and the confirmed placement or execution of an order. \gao{With this observation, we now define the value of an order mathematically. Suppose
based on the information at time zero, an order with relative price $x\in \mathbb{Z}$ is sent to the exchange. For an ask order, this means the order is
quoted at the price $p(0)+0.5+x$; For a bid order, this means the order is quoted at the price $p(0)- 0.5 - x$. }
This order experiences a time delay $t_1' \ge 0$ (this is the latency for the new order sent by the market maker)
 before its placement is confirmed by the exchange. We then compare the execution price of this order with the mid-price at time $t_1' + t_2'$ (for some time $t_2' \ge 0$) which is regarded as the fundamental value of the asset. If the order is not executed by the time $t_1' + t_2'$, then the value of this order is zero. As we will see later, in our model, $t_1'$ can be zero or $\Delta \tau$ and $t_1' + t_2'$ can be  $\Delta \tau$ or $\Delta t$ depending on the order that we consider (e.g. an outstanding old order or a new order sent by the maker) and when the order is canceled.

\rev{Mathematically, for any $t_1',t_2'\geq 0$, the value of an order with relative price $x$ is defined by
\begin{equation}\label{H(xy)}
H^\pm(t_1',t_2',x)=E\left[ \left(\max\{0.5+x\mp\Delta p[0,t_1'],-0.5\}\mp \Delta p[t_1',t_1'+t_2'] \right) \cdot \mathds{1}_{fill^\pm_{t_1',t_2',x}}\right],
\end{equation}
where $H^+$ denotes the value of an ask order, $H^-$ denotes the value of a bid order, $\Delta p[\cdot, \cdot]$ denotes the change of mid-price on the corresponding time window, and the indicator functions $\mathds{1}_{fill^\pm_{t_1',t_2',x}}$ specify whether the ask/bid order with relative price $x$ is filled before time $t_1'+t_2'$.}

\rev{We briefly explain \eqref{H(xy)} for the ask side as the bid side is similar. The execution price of this ask order is its limit price $p(0)+0.5+x$, or the market best bid price at the time $t_1'$, i.e., $p(0)+\Delta p[0,t'_1]-0.5$, depending on whether the market best bid exceeds the limit price of the ask order when the ask order reaches the exchange. In addition, the mid-price at time $t_1' + t_2'$ is given by $p(0) + \Delta p[0,t_1'] +  \Delta p[t_1',t_1'+t_2'] $. The difference between the execution price of the ask order and the mid price at time $t_1' + t_2'$ yields the expression of order value $H^+$ in \eqref{H(xy)}. Note that the above definition of order value $H^\pm(t_1',t_2',x)$ also applies to $x=\pm\infty$. For $x=\infty$, the indicator functions in \eqref{H(xy)} is zero since such an order will not be filled, and $H^\pm(t_1',t_2',\infty)\equiv0$. Similarly, $x=-\infty$ means the order is a market order, and we have $H^\pm(t_1',t_2',-\infty)\equiv-0.5$ by the martingale property of the mid-price.}

\section{Main results}\label{sec theoretical}
\noindent
In this section, we present the main results. The first result is on the value of an order and it is given in Section~\ref{sec:ordervalue}. The second result characterizes the structure of the value functions of the MDP model and it is provided in Section~\ref{sec:structure-vf}. Finally, the third result is on when the market making strategy is profitable and it is given in Section~\ref{sec:profit-mm}.

\subsection{Characterization of the value of an order} \label{sec:ordervalue}
We now state the first main result of this paper. It gives an explicit expression of the value of an order with delay $t_1'=0$. It also connects the value of an order experiencing positive delay with that of zero delay. Finally, it provides simple conditions to decide when the order value is positive. Recall $H^{\pm}$ and the indicator functions $\mathds{1}_{fill^\pm_{0,t_2',x}}$ defined in \eqref{H(xy)}. With slight abuse of notations, we use $\delta^\pm$ to denote the relative prices of ask/bid orders.

\begin{theorem}\label{value of orders tau}
	For any $t_1'\geq0$, $t_2'>0$ and $(\delta^+,\delta^-)\in\mathbb{\overline{Z}}^2$, we have
\begin{itemize}
\item [(a)]
	\rev{\begin{equation} \label{eq:value1}
	H^\pm(0, t_2',\delta^\pm)=
	\left\{
	\begin{split}
	&\left(\frac{\lambda^\pm}{\lambda^\pm+\lambda/2}-0.5 \right) \cdot E \left[\mathds{1}_{fill^\pm_{0,t_2',\delta^\pm}} \right],\quad  &\delta^\pm\ge 0,\\
	&-0.5,\quad &\delta^\pm< 0,
	\end{split}
	\right.
	\end{equation}}
\item [(b)]
	\rev{\begin{equation}
	H^\pm(t_1',t_2',\delta^\pm)=E[H^{\pm}(0, t_2',\delta^\pm \mp \Delta p[0,t_1'])]
	\end{equation}}

\item [(c)] \rev{$H^{\pm}(t_1',t_2',\delta^\pm)\leq 0$ for any $\delta^\pm\in\mathbb{\overline{Z}}$ if and only if
$\lambda^\pm\leq \lambda/2$.}
\end{itemize}
\end{theorem}

Part (a) of this result states that the value of an order with zero delay is linear in its fill probability. We briefly discuss the economic interpretations of the coefficients.
For illustrations, we take \rev{$H^{+}(0,t_2',\delta^+)$} as an example. For \rev{$\delta^+< 0$}, the ask order is effectively a market order, and will be filled instantly at the current best bid price as there is no latency. So its execution price is 0.5 ticks less than the mid-price at the time of execution. As the mid-price is a martingale with independent increments, it follows that the expected profit or the value of such an order is $-0.5$ ticks. On the other hand,
for \rev{$\delta^+ \ge 0$}, the limit sell order enters into the order book, and by \eqref{eq:value1} its value equals to a constant \rev{$\frac{\lambda^+}{\lambda^++\lambda/2}-0.5$} multiplied by the fill probability of the ask order \rev{in $[0, t_2']$, i.e.,  $E[\mathds{1}_{fill^+_{0,t_2',\delta^+}}]$}. This constant \rev{$\frac{\lambda^+}{\lambda^++\lambda/2}-0.5$} represents the conditional expected profit of the ask order given that the order is filled. To see this, we note that when such an ask order is filled, there are two scenarios: first, the ask order sits at the best ask price, and eventually transacts with an `uninformed' buy order, gaining 0.5 tick as the mid-price does not move at the time of execution (`spread capture'); second, the \rev{mid-price} jumps up and crosses the quoted price of the ask order, in which case, the ask order loses 0.5 tick as the mid-price immediately moves up one tick at the time of execution of the order (`adverse selection'). The rate of the first scenario occurs is \rev{$\lambda^+$}, while the rate of the second scenario occurs is $\lambda/2$. Hence, the expected profit the limit sell order conditioned on execution is given by
\rev{\begin{equation*}
\frac{\lambda^+}{\lambda^++\lambda/2} \cdot 0.5+\frac{\lambda}{\lambda^++\lambda/2} \cdot (-0.5)=\frac{\lambda^+}{\lambda^++\lambda/2}-0.5.
\end{equation*}}
While our model does not feature information asymmetry, the result here is
 consistent with the existing studies (see e.g., \cite{sandas2001, Moallemi2}) where they showed that one can interpret the order value as follows:
\[\text{Value of an order} =\text{(spread capture $-$ adverse selection cost)} \times \text{fill probability}.\]

We next discuss Part (b). It suggests that the value of orders with latency $t_1'$ is the expected value of orders with zero latency where the quotes \rev{$(\delta^+, \delta^-)$} are perturbed by random fluctuations of the market price during the latency window. So if, for example, quoting at the best ask price \rev{$\delta^+=0$} is optimal (in maximizing the order value) in the case of zero delay $t_1'=0$, then such a quote may be not optimal anymore in the case of a delay $t_1'>0$, and the order value will be decreased.

Part (c) says if \rev{$\lambda^+\le\lambda/2$}, then there are no ask orders with positive values and if \rev{$\lambda^+>\lambda/2$}, there is at least one ask limit order whose value is positive. It is similar for the bid side. This result will be useful in characterizing the profitability of market making strategies discussed in Section~\ref{sec:profit-mm}. Note that for $t_1'=0$, the result in Part (c) clearly holds in view of Part (a). For general $t_1' \ge 0$, one needs to combine Parts (a) and (b), and properties of the random walk and Poisson processes to establish the result. See Appendix~\ref{sec:proofOrderValue} for details of the proof.

\subsection{Structure of value functions and the role of order value} \label{sec:structure-vf}
We now present the second main result of the paper. It characterizes the structure of value functions in \eqref{Bellman} and highlights the role of value of orders in optimal market making.

To facilitate and simplify the presentation, let us use the indicator functions
\rev{$\mathds{1}_{fill^+_0}$ and $\mathds{1}_{fill^-_0}$} to denote whether the ask and bid outstanding  orders (observed at time 0) are filled in $[0, \Delta \tau$). We also define
\rev{\begin{equation}\label{Hask act}
	H^{\pm}_{act}(r^\pm, \delta^\pm):=
	\left\{
	\begin{aligned}
	&H^{\pm}(0,\Delta \tau,r^\pm)+H^{\pm}(\Delta\tau,\Delta t-\Delta\tau,\delta^\pm), &\delta^\pm\in\mathbb{\overline{Z}},\\
	&H^{\pm}(0,\Delta t,r^\pm), &\delta^\pm=o,\\
	\end{aligned}
	\right.
	\end{equation}}
Then we can obtain the following result. It shows
 that the value of the market maker's bid and ask orders in one period, that is \rev{$H^{+}_{act}(r^+, \delta^+)+H^{-}_{act}(r^-, \delta^-)$} as we will explain later, essentially plays the role of one-period reward in our MDP model.

\rev{\begin{theorem}\label{BellmanThm}
	For any $s=(w,p,q,r^+,r^-)\in S$, we have $v_{N+}(s)=w+pq-0.5|q|$ and
	\begin{equation}\label{Bellman v_i}
	v_i(s)=w+pq+g_i(q,r^+,r^-), \quad i=0,1,2,...,N,
	\end{equation}
	where
	\begin{equation}\label{Bellman g_i}
	g_i(q,r^+,r^-):=
	\left\{
	\begin{aligned}
	&\begin{split}
	&H^{+}(0,\Delta\tau,r^+)+H^{-}(0,\Delta\tau,r^-)\\
	& \qquad-0.5E \left[\;|q-\mathds{1}_{fill^+_0}+\mathds{1}_{fill^-_0}|\; \big| (r^+_0,r^-_0)=(r^+,r^-) \right],
	\end{split} & i=N,\\
	&\max\limits_{(\delta^+,\delta^-)\in A_s} G_i(q,r^+,r^-,\delta^+,\delta^-), & i=0,1,...,N-1,
	\end{aligned}
	\right.
	\end{equation}
and
	\begin{equation}\label{H_i 5}
	\begin{split}
	&G_i(q,r^+,r^-,\delta^+,\delta^-):=E \left[g_{i+1}(q_1,r^+_1,r^-_1) \mid (q_0,r^+_0,r^-_0)=(q,r^+,r^-),(\delta^+_0,\delta^-_0)=(\delta^+,\delta^-)\right]\\
& \qquad \qquad \qquad \qquad  \qquad  + H^{+}_{act}(r^+, \delta^+)+H^{-}_{act}(r^-, \delta^-).
	\end{split}
	\end{equation}
\end{theorem}}


Theorem \ref{BellmanThm} reduces the computation of value functions from five state-variables to three state-variables \rev{$(q, r^+, r^-)$} in the backward recursion \eqref{Bellman g_i}.
As suggested by \eqref{Bellman v_i}, the value function \rev{$v_i(w,p,q,r^+,r^-)$} can be decomposed into three parts: (1) $w$ represents the market maker's current wealth or cash; (2) \rev{$pq$} is the value of the inventory marked to the market at the mid-price; (3) \rev{$g_i(q,r^+,r^-)$} represents the extra value from following the optimal strategy, and this extra value depends on the market maker's outstanding orders as well as the inventory. In particular, there is an inventory penalty $-0.5|q|$ at the terminal time \rev{$t_{N+}$} which arises as the market maker needs to unwind his positions by using a market order and crossing the bid-ask spread.


The backward recursion and maximization problem in \eqref{Bellman g_i} and \eqref{H_i 5} specify the trade-off between the value of orders (outstanding orders and new orders sent in the action) in the current period and the expected extra value $g_{i+1}$ at the next period. Specifically,  the term \rev{$H^{+}_{act}(r^+, \delta^+)+H^{-}_{act}(r^-, \delta^-)$} in \eqref{H_i 5} is the value of outstanding orders and new orders in the actions in the current period. To see this, note if \rev{$\delta^+,\delta^-\in \mathbb{\overline{Z}}$}, then the lifetime of the outstanding orders in the current period is $\Delta\tau$; otherwise (i.e., doing nothing) the outstanding order stays in the order book for $\Delta t$ in one period. Similarly, the new orders sent in the $i$-th action, if entered into the order book, stays for $\Delta t-\Delta\tau$ in the current period.
This explains \eqref{Hask act} and hence the interpretation for \rev{$H^{+}_{act}(r^+, \delta^+)+H^{-}_{act}(r^-, \delta^-)$}.

We also explain the similarities and differences between the structure of the value functions here and that in the existing literature on optimal market making with zero latency (under different models). When there is zero latency, i.e., $\Delta \tau =0$, one can readily show that the value function does not depend on the relative prices of outstanding orders \rev{$(r^+,r^-)$}. In this case, the decomposition structure of value functions in \eqref{Bellman v_i} is similar as in \cite{Cartea}. Due to the presence of latency in our model, we can observe two main differences between our result and those in the literature (see e.g., \cite{Cartea}). First, our extra value $g_i$ depends on the outstanding orders.
Second, due to the existence of latency, these outstanding orders may be executed before the market maker can cancel them, which will affect the future inventory and actions of the market maker.

\subsection{Profitability of market making strategies} \label{sec:profit-mm}

We present the third main result of the paper. It allows us to better understand the profitability of the market making strategies with $\Delta\tau\geq0$. To be specific, we consider \rev{the expected profit $\mathcal{P}$} of the market maker, defined as
\rev{\begin{equation}\label{NP 1}
\mathcal{P}:=v_0(w,p,0,\infty,\infty)-w.
\end{equation}}
That is, \rev{$\mathcal{P}$} is the expected net wealth change over the horizon $[0, T]$, where the market maker starts with cash $w$, zero inventory ($q=0$) and no outstanding orders \rev{($r^+=\infty, r^- =  \infty$)} in the limit order book at time zero. Our result given below shows when \rev{$\mathcal{P}$} is positive.

\begin{theorem}\label{net profit} Fix latency $\Delta \tau \ge 0$ and other
parameters $\lambda,\rev{\lambda^+,\lambda^-}, \Delta t, \overline{q},\underline{q}$. We have
	\begin{itemize}
		\item [(1)] If \rev{$\lambda^+\leq \lambda/2$ and $\lambda^-\leq \lambda/2$} , then \rev{$\mathcal{P}=0$} for any $N \ge 1$.

\item [(2)] If \rev{$\lambda^+ > \lambda/2$ and $\lambda^- > \lambda/2$}, then there exists a finite positive integer $N_{min}$ depending on the fixed parameters such that \gao{$\mathcal{P}=0$ for $N < N_{\min}$ and $\mathcal{P}>0$ for $N \ge  N_{\min}$. }
	\end{itemize}
\end{theorem}

We next discuss the implications of Theorem~\ref{net profit}.

Part (1) of this result says that under the conditions that the rates of `uninformed' market orders that transact with market maker's limit orders \rev{$\lambda^+, \lambda^-$} are smaller than $\lambda/2$, then the market maker can not make a positive profit regardless of how many times the market maker quotes. These conditions are likely to hold when either the market is highly volatile with a large rate of price change $\lambda$ or there are not sufficient `uninformed' market orders hitting the market maker's limit orders which leads to low values of \rev{$\lambda^+$ and $\lambda^-$}. Part (1) then suggests that in these scenarios, electronic market making on the single asset is not profitable.

Part (2) of this result suggests that the market making strategy can be profitable if the market conditions are good in the sense that \rev{$\lambda^+, \lambda^- > \lambda/2$},  and the market maker can quote a large number of times, or equivalently, have a long trading horizon when the period length $\Delta t$ is fixed. Though we do not have an analytical expression for ${N}_{\min}$, this quantity has an explicit computable upper bound which only involves the values and fill probabilities of certain ask and bid orders. See the proof of Theorem~\ref{net profit} for details.

To understand the intuition behind Theorem~\ref{net profit}, we use the characterization of the order value in Theorem \ref{value of orders tau} and the value functions in Theorem \ref{BellmanThm}. On the one hand, by Part (c) of Theorem \ref{value of orders tau}, there are no orders whose order values are positive when \rev{$\lambda^+\leq \lambda/2$ and $\lambda^-\leq \lambda/2$}. As Theorem \ref{BellmanThm} suggests that the order value essentially plays of the role of one-period reward in the dynamic optimization problem, we can infer that the market maker can not earn a positive profit. On the other hand,
when \rev{$\lambda^+, \lambda^- > \lambda/2$}, Part (c) of Theorem \ref{value of orders tau} shows that there exists an ask-bid quote pair whose order values are both positive. Then we can use this quote pair to construct an explicit feasible market making policy, for which the profit becomes positive for sufficiently large number of quoting times $N$.

Before we proceed to numerical experiments, we briefly comment on the cases \rev{$\lambda^+>\lambda/2\ge \lambda^-$ and $\lambda^->\lambda/2\ge \lambda^+$}. In both cases, there exist orders with positive order value for one side while there are none for the other side. Depending on the values of \rev{$\lambda^+, \lambda^-$} and other model parameters, it is possible that \rev{$\mathcal{P}=0$} for any $N$ (as in Part (1) of Theorem~\ref{net profit}) or
\rev{$\mathcal{P}>0$} for sufficiently large $N$ (as in Part (2) of Theorem~\ref{net profit}).

\section{Numerical experiments} \label{sec: numerical results}
In this section, we present numerical results.
Section~\ref{sec:estimation} discusses estimations of model parameters using NASDAQ data. Section~\ref{simulation} discusses a representative example of the optimal quoting policy of the market maker and the associated inventory process.
Section~\ref{sec:ordervalue-latency} illustrates our results on order value. Section~\ref{sec:profit-latency} illustrates our main results on profitability of market making and effect of latency numerically. Our numerical results are based on solving the backward recursion in Theorem \ref{BellmanThm}, where we can compute the functions $g_i$, and find the optimal quotes by truncations of the infinite state and action spaces and using exhaustive search for the maximization problem in \eqref{Bellman g_i}.


\subsection{Estimations} \label{sec:estimation}
\noindent
We discuss the estimations of model parameters $\lambda,\rev{\lambda^+, \lambda^-}$ in this section. Other parameters such as the quote duration $\Delta t$, the trading horizon $T$, the inventory bounds $\underline{q},\overline{q}$, are all chosen by the market maker.

\subsubsection{Data description}
\noindent
 We use NASDAQ's TotalView-ITCH data, which contains message data of order events.\footnote{Data is provided by LOBSTER website (https://lobsterdata.com/).} The database documents all the order activities causing an update of the limit order book up to the requested number of levels and thus includes visible orders' submissions, cancellations and executions with order reference numbers. Each visible limit order is identified with a unique order reference number which is assigned immediately after the submission. The timestamp of these events is measured in seconds with decimal precision of at least milliseconds and up to nanoseconds depending on the requested number of levels. 

We conduct the parameter estimation using one representative large-tick stock, General Electric Company (GE), with data from 10:00 a.m. to 4:00 p.m. on a randomly selected day: Oct 3, 2016.
Table \ref{ticksize} shows the observations of the bid-ask spreads on that day. As one can see, the spread of GE is 1 tick for the most of the time and the spread is rarely larger than 2 ticks.
We also report that the average sizes of limit order queues on best ask and best bid are 8758 and 7727 shares respectively.

\begin{table}[h]
	\begin{center}
		\caption{Percentages of observations with different bid-ask spreads for GE on Oct 3, 2016.}
		{\begin{tabular}{@{}lccc}
				\hline
				Bid-ask spread  &1 tick & 2 ticks  & $\geq$ 3 ticks\\
				Percentage& 88.236   & 11.757 & 0.007\\
				\hline
		\end{tabular}}
		\label{ticksize}
	\end{center}
\end{table}

\subsubsection{Estimations of model parameters}\label{sec: estimates}
\noindent
We first discuss how to estimate $\lambda$, the intensity of market price change in our model.
As we assume in the model that the bid-ask spread is always 1 tick, we first delete the data when the bid-ask spread is more than 1 tick.
Then, we can estimate the intensity of price change using the average number of jumps of the mid-price per minute during the day. This yields
\begin{equation*}
\lambda=\lambda_{GE}=1.56 \quad \text{ per minute.} 
\end{equation*}

We then discuss how to estimate \rev{$\lambda^+, \lambda^-$}, the rates of `uninformed' market orders that transacts with the market maker's limit orders at best quotes.
We count a market order as an `uninformed' market order if the mid-price does not change after the market order arrives and generates trades, and then we estimate the intensity of arrivals of the total `uninformed' market orders by computing the average numbers of arrivals per minute. For simplicity, we use half of the total intensity for buy and sell `uninformed' market orders, which leads to
\begin{equation*}
\rev{\lambda^+_{GE}=\lambda^-_{GE}}=1.25 \quad \text{ per minute.} 
\end{equation*}
Note that \rev{$\lambda^+_{GE}$ and $\lambda^-_{GE}$} correspond to the case in which a market maker's orders are always at the front of the queue at the best quotes.
For a particular market maker, the rates \rev{$\lambda^+$ and $\lambda^-$} could be less than \rev{$\lambda^+_{GE}$} as the orders he sent may not have the highest execution priority in the queue. In addition, \rev{$\lambda^+$ and $\lambda^-$} may depend on other factors including the latency and the speed advantage of a particular market maker compared with traders performing similar strategies.
For illustration purposes, in our numerical experiments, we fix \rev{$\lambda^+= \lambda^-= 0.7 \cdot \lambda^+_{GE} = 0.875$ per minute $>0.5\lambda_{GE}$} unless otherwise specified. In particular, we assume \rev{$\lambda^+, \lambda^-$} are independent of the latency $\Delta \tau$ in this section.
We remark that it is possible to estimate \rev{$\lambda^+$ and $\lambda^-$} and show their dependence on $\Delta \tau$ for a particular market maker
based on our theoretical analysis in the paper. See Appendix~\ref{sec:est-lambda-a} for details.

\subsection{The market maker's optimal quotes and inventory process}\label{simulation}

Based on the parameters estimated above, we show in Figure \ref{revsimulation} a representative example of the market maker's optimal quotes and the inventory process in one simulation. In this example, the market maker quotes every one second in an 100-minute window. The latency is fixed at 0.02 seconds. In this sample path,
the market maker sends 129 ask orders and 35 bid orders in total. Among all these orders, 13 ask orders and 13 bid orders are executed.

\rev{In the optimal quoting strategy,
we use quote value 15 to denote the action $\delta^+=o$, i.e., do nothing on the ask side. Similarly, we use
quote value -14 to denote $\delta^-=o$. For any other value $y\in (-14,15)$ in the left vertical axis, it represents the market maker's quote $\delta^+=y-1$ for the ask side or $\delta^-=-y$ for the bid side.} For this particular sample path, we observe that the action of the maker is doing nothing for many periods. This is because sending new orders have potential drawbacks due to latency: if the market price moves in the latency window, then the new orders sent by the maker may enter into undesired price levels or get `adversely filled' when arriving at the order book. So if the maker already has outstanding orders at suitable prices, then it is possible that doing nothing is better than sending new orders to replace the outstanding orders.
We also observe that when the maker has nonzero inventories, his quote pair may be asymmetric. For example, when time is 450 seconds, \rev{the inventory is -2 (the right axis) and the action is $(\delta^+,\delta^-)=(4,0)$}. In this case, the market maker sells less aggressively given that the inventory is negative.
\begin{figure}[h]
	\centering
	\includegraphics[width=0.75\textwidth]{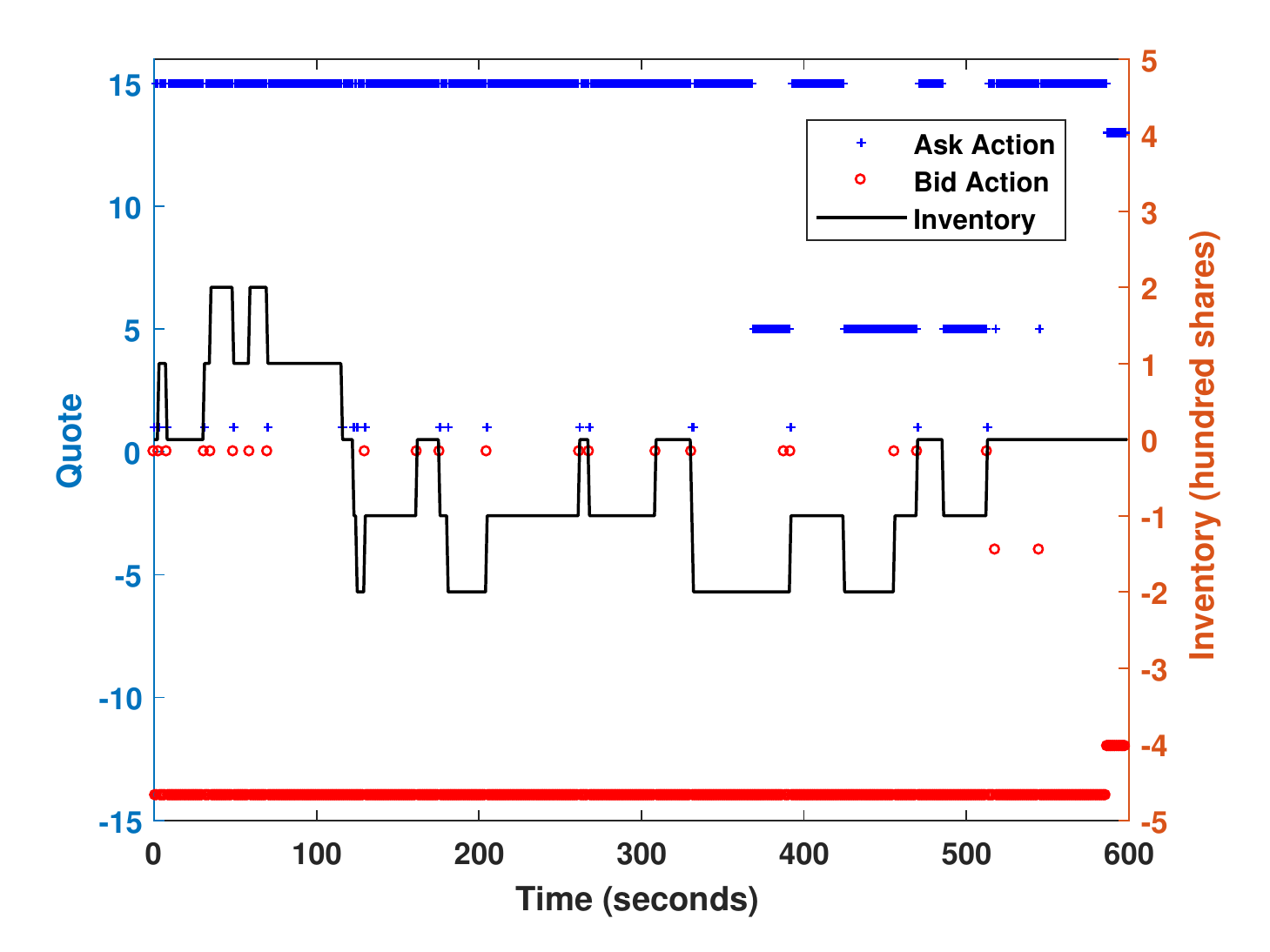}
	\caption{\label{revsimulation} The market maker's optimal quotes and the inventory process in one simulation. The model parameters are: $\underline{q}=-4, \overline{q}=4, \Delta \tau=0.02$ seconds, $\Delta t= 1$ second, $T=600$ seconds, $\lambda=\lambda_{GE}$ and \rev{$\lambda^+=\lambda^-=0.7 \lambda^+_{GE}$}.}
\end{figure}

\subsection{Order Value and Latency} \label{sec:ordervalue-latency}
\noindent
In this section, we use numerical experiments to illustrate our results in Sections~\ref{sec:ordervalue} on order value
based on the parameters from Section~\ref{sec: estimates}.

\rev{Table \ref{revordervaluetable} shows, using two representative examples with different latencies, the value of new ask orders sent by the maker \rev{$H^{+}(\Delta\tau,\Delta t-\Delta\tau,\delta^+)$} (see \eqref{H(xy)} and \eqref{Hask act}) for different quote decisions \rev{$\delta^+$}.} We only show the order value for \rev{$\delta^+$ from 0 to 4}, because the order value is almost zero for \rev{$\delta^+>4$} (as the order fill probability is very small) and the order value is close to $-0.5$ for \rev{$\delta^+<0$} (as the ask order will be filled instantly with a high probability when arriving at the order book).

We can make two observations from \rev{Table \ref{revordervaluetable}}. First, for both cases in \rev{Table \ref{revordervaluetable}}, there exists \rev{$\delta^+$} such that the order value is positive. This is consistent with Part (c) of Theorem \ref{value of orders tau} as \rev{$\lambda^+=\lambda^-> \lambda/2$} in \rev{Table \ref{revordervaluetable}}. Second, the maximum order value is attained at \rev{$\delta^+=0$} for $\Delta \tau=0.2$ seconds, which is $1.48 \times 10^{-3}$, and at \rev{$\delta^+=1$} for $\Delta \tau=0.8$ seconds, which is $4.58\times 10^{-5}$. This is not surprising in view of Theorem \ref{value of orders tau}. On the one hand,
when $\Delta\tau=0$, by Part (a) of Theorem \ref{value of orders tau}, the order value equals to a constant multiplied by the order fill probability and hence it attains the maximum when \rev{$\delta^+=0$}. Then intuitively one expects this also holds for small $\Delta\tau$ by continuity. On the other hand, when $\Delta\tau$ becomes larger, the probability that the ask order quoted at \rev{$\delta^+=0$} is filled instantly (`adversely fill') when arriving at the order book increases, due to the possible price movement in the latency window. Thus, order value at \rev{$\delta^+=0$} may become negative and hence the maximum value may be attained at a larger \rev{$\delta^+$}. 

\begin{table}[h]
	\begin{center}
		\caption{Value of new ask orders \rev{(measured in dollars)} \rev{$H^{+}(\Delta\tau,\Delta t-\Delta\tau,\delta^+)$} for different \rev{$\delta^+$} and latency $\Delta\tau$. The remaining model parameters are: $\Delta t=4$ seconds, and $\lambda=\lambda_{GE}, \rev{\lambda^+=\lambda^-=0.7\lambda^+_{GE}}>0.5\lambda_{GE}$.}
	\footnotesize{\begin{tabular}{@{}lccccc}
				\hline
				
				$\delta^+$  & 0 & 1 & 2 & 3 &4\\
				order value($\Delta\tau=0.2$ seconds) & $1.48 \times 10^{-3}$   & $7.36 \times 10^{-5}$ & $1.31 \times 10^{-6}$ & $1.91\times 10^{-8}$  & $1.27 \times 10^{-9}$ \\
				order value($\Delta\tau=0.8$ seconds) & $-2.79 \times 10^{-3}$   & $4.58 \times 10^{-5}$ & $1.20 \times 10^{-6}$ & $1.76\times 10^{-8}$  & $ 6.49 \times 10^{-10}$ \\
				\hline
		\end{tabular}}
		\label{revordervaluetable}
	\end{center}
\end{table}

Figure \ref{revoutordervaluepic} shows, the value of ask outstanding orders \rev{$H^{+}(0,\Delta\tau,r^+)$} (see \eqref{H(xy)} and \eqref{Hask act}) with $r^+=0, 1$, as a function of latency $\Delta\tau$, in an unfavorable situation:
\rev{$\lambda^+=\lambda^-<0.5\lambda$}. Note that the order values are all negative which is consistent of part (a) of Theorem \ref{value of orders tau}. The fill probability of an outstanding order increases as latency increases, and hence by part (a) of Theorem \ref{value of orders tau}, the order value decreases. Note the value of outstanding orders plays a important role in the value function (see \eqref{Hask act}). Thus, this experiment shows another disadvantage of latency for the market maker: if the maker cannot cancel his stale orders in a timely fashion due to latency, then it will decrease his profit when the market condition is unfavorable.

\begin{figure}
	\centering
	\includegraphics[width=0.75\textwidth]{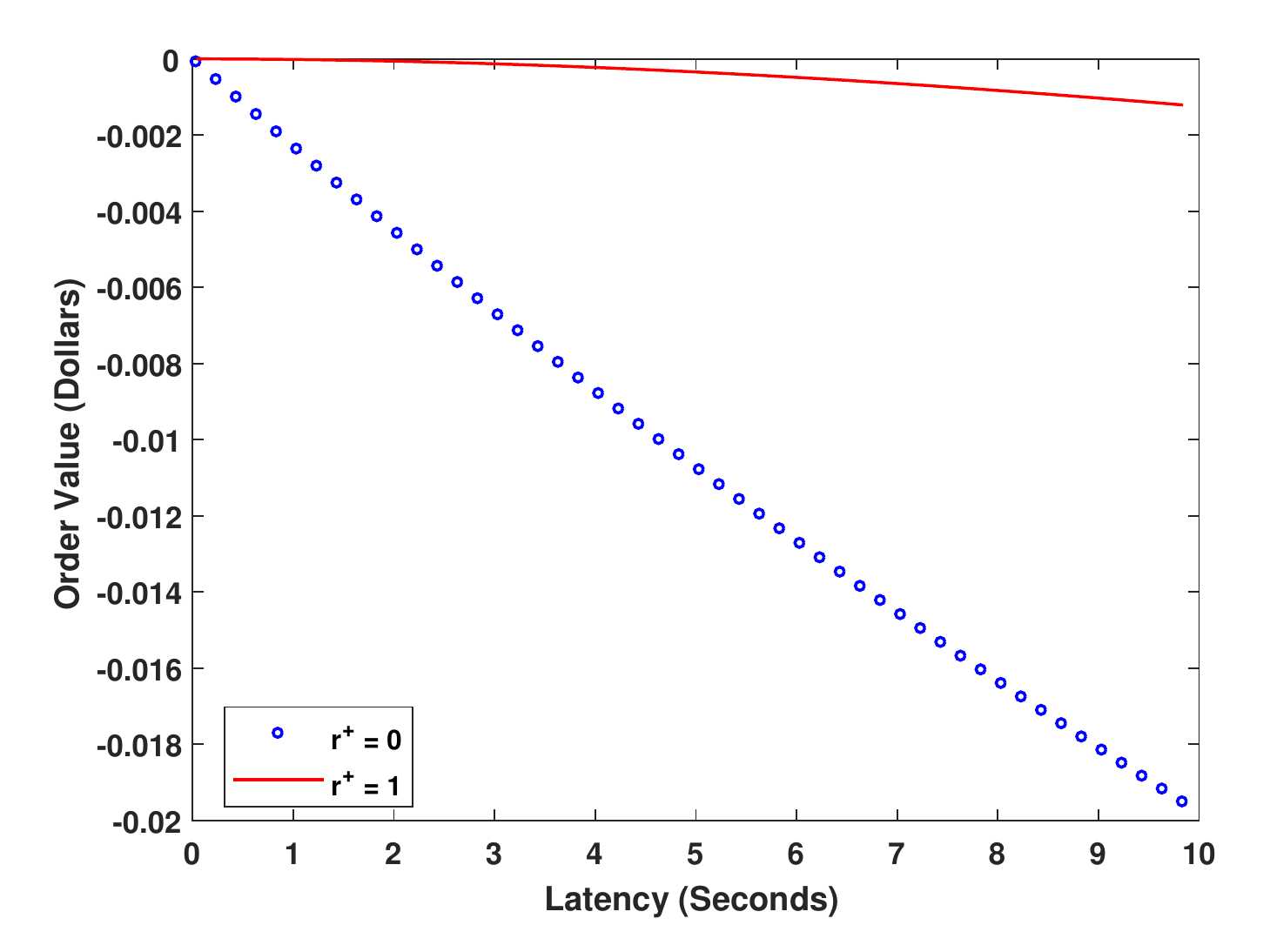}
	\caption{\label{revoutordervaluepic} Value of ask outstanding orders \rev{$H^{+}(0,\Delta\tau,r^+)$} as a function of $\Delta\tau$ (measured in seconds) for different \rev{$r^+$}. The remaining model parameters are: $\lambda=\lambda_{GE}$ and
		$\rev{\lambda^+=\lambda^-=0.4\lambda^+_{GE}}<0.5\lambda_{GE}$.}
\end{figure}

\subsection{\rev{Expected Profit} of Market Making and Latency} \label{sec:profit-latency}
In this section, we illustrate numerically our results in Section~\ref{sec:profit-mm} on the profitability of market making and effect of latency.

We show in Figure \ref{revprofitandlatency1}, using two representative examples with \rev{$\lambda^+=\lambda^-=0.7\lambda^+_{GE}>0.5\lambda_{GE}$ } and \rev{$\lambda^+=\lambda^-=0.5\lambda_{GE}$}, the market maker's \rev{expected profit $\mathcal{P}$} as a function of latency $\Delta\tau$. There are two immediate observations from Figure \ref{revprofitandlatency1}. First, we find that when \rev{$\lambda^+=\lambda^-= 0.5\lambda_{GE}$}, \rev{$\mathcal{P}$} are always zero for any $\Delta \tau$. This is also consistent with Part (1) of Theorem~\ref{net profit}. Second, fixing other parameters, \rev{$\mathcal{P}$} is a non-increasing function of $\Delta\tau$. In particular, when the latency
$\Delta\tau$ is large, then \rev{$\mathcal{P}$} becomes zero. This is because the number of quotes $N=599$ is fixed in this example, while numerically the threshold $N_{\min}$ for earning positive profit in Theorem~\ref{net profit} increases with $\Delta\tau$. So when the latency is large, we have $N < N_{\min}$ and the \rev{expected profit} of the market maker is zero by Part (2) of Theorem~\ref{net profit}.

\begin{figure}
	\centering
	\includegraphics[width=0.75\textwidth]{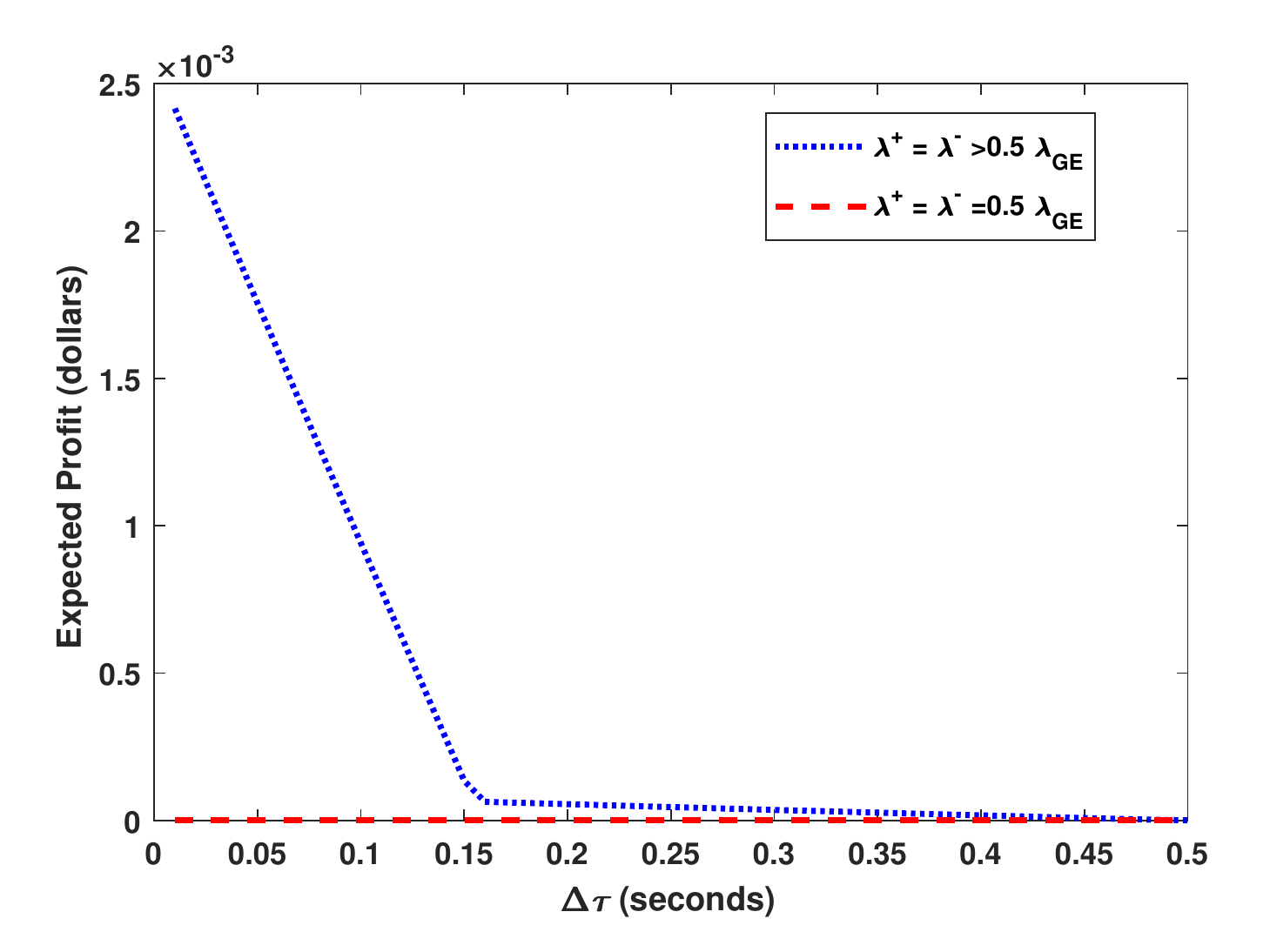}
	\caption{\label{revprofitandlatency1} Net Profit \rev{$\mathcal{P}$} of the market maker as a function of $\Delta\tau$ for different \rev{$\lambda^+$ $(=\lambda^-)$}. The remaining model parameters are: $\underline{q}=-2, \overline{q}=2, \Delta t=0.5$ seconds, $T=90$ seconds, and $\lambda=\lambda_{GE}$.}
\end{figure}

\section{Concluding remarks and future research} \label{sec:conclusion}

In this paper, we propose a stylized discrete-time MDP model to study
optimal market making strategies for large--tick assets in the presence of latency. We provide explicit characterizations of the order value, as well as the structure of the value functions. We use such characterizations to study the profitability of market making strategies and the effect of latency. Our analysis suggests that latency can be an additional risk source for market makers and one key task in market making is to predict the order values based on market primitives, see, e.g., \cite{Moallemi2} for an example of recent developments along this line. For future research to extend our model in the current paper, a variety of realistic features can be considered. For example, the latency a market maker experiences in practice is random and may vary when market condition changes, e.g.
due to a huge increase in quotes and cancellation requests. In addition, high-frequency asset price typically demonstrates autocorrelations in real markets. We hope that our model stimulates further research to address such practical issues.

\section*{Acknowledgements}
We thank Sebastian Jaimungal, Steve Kou, and two anonymous reviewers for their suggestions and comments. This research is supported by Hong Kong RGC Grants 24207015 and 14201117.

\clearpage
\appendix


\setcounter{equation}{0}
\renewcommand{\theequation}{A\arabic{equation}}
\section{Further details of the Model in Section~\ref{model}}
\subsection{The mathematical expression of the admissible action space} \label{sec:As}
We adopt the following conventions related to the operations for $\pm\infty$: (1) $-\infty<x<\infty$ for any $x\in\mathbb{Z}$; (2) $\infty+x=\infty$ for any $x\in\mathbb{Z}\cup{\{\infty\}}$; (3) $-\infty+x=-\infty$ for any $x\in\mathbb{Z}\cup{\{-\infty\}}$; and (4) for any $x\in\mathbb{Z}\cup{\{\pm\infty\}}$, $\pm\infty \times x = \infty$ if $x>0$; $\pm\infty \times x = -\infty$ if $x<0$;
$\pm\infty \times x =0$ if $x=0$.

We now give the expression of the market maker's admissible action space $A_s$ for state \rev{$s=(w,p,q,r^+,r^-)\in S$}.
\rev{We first define two subsets of $S$ in which the inventory of the maker may reach the bounds.
\begin{equation}\label{Sdown}
\underline{S} :=
\left\{
\begin{aligned}
&\{(w, p, q, r^+,r^-) \in S: q=\underline{q} \}, &\Delta \tau=0,\\
&\{(w, p, q, r^+,r^-)\in S: q=\underline{q} \text{ or } q=\underline{q}+1, r^+<\infty \}, &\Delta \tau>0.\\
\end{aligned}
\right.
\end{equation}
The set $\underline{S}$ contains the states in which the market maker's inventory has either reached the lower bound $\underline{q}$ or will reach the lower bound if the ask outstanding order gets filled and the bid order does not get filled. Note that the outstanding ask orders will get canceled instantly when a new ask order is sent if latency $\Delta\tau=0$.
In these cases, the market maker should not quote new ask orders or he should use a buy market order in order for the inventory to stay within the bound, i.e., $\delta^+\in \{\infty,o\}$ or $\delta^-=-\infty$. Similarly we define a set  $\overline{S}$ for the upper inventory bound $\overline{q}$ as follows.
\begin{equation}\label{Sup}
\overline{S} :=
\left\{
\begin{aligned}
&\{(w, p, q, r^+,r^-) \in S: q=\overline{q} \}, &\Delta \tau=0,\\
&\{(w, p, q, r^+,r^-)\in S: q=\overline{q} \text{ or } q=\overline{q}-1, r^-<\infty \}, &\Delta \tau>0.\\
\end{aligned}
\right.
\end{equation}
Then, the admissible action space for state $s$ is given by
\begin{equation}\label{A}
	A_s:=
	\left\{
	\begin{aligned}
	&\{\infty,o\}\times \mathbb{\overline{Z}}^o\cup\mathbb{\overline{Z}}^o\times\{-\infty\} , \quad &s\in \underline{S},\\
	&\{-\infty\}\times \mathbb{\overline{Z}}^o\cup\mathbb{\overline{Z}}^o\times\{\infty,o\} , \quad &s\in \overline{S},\\
	&\mathbb{\overline{Z}}^o\times\mathbb{\overline{Z}}^o, \quad &\text{otherwise}.
	\end{aligned}
	\right.
\end{equation}}

\subsection{System dynamics for our MDP model}\label{sec:dynamics}
\noindent
We now describe the dynamics of the discrete system states of the MDP, i.e., $s_i$, $i=0,1,2,...,N,\rev{N+}$. For $i=0,1,...,N-1$, denote the $i$-th action/decision of the maker by \rev{$(\delta^+_i,\delta^-_i)$}.
For period $i=0,1,...,N$, we use two indicator functions \rev{$\mathds{1}_{fill^+_i}$ and $\mathds{1}_{fill^-_i}$} to specify whether the ask and bid outstanding orders (if any) at time $t_i-$ are filled (1 if filled; 0 if not) before time \rev{$t_{i+}$} respectively. Note these outstanding orders will be canceled at time \rev{$t_{i+}$} if the market maker sends cancellation instructions in the $i$-th action.
Similarly, for $i=0,1,2,...N-1$, we use two indicator functions \rev{$\mathds{1}_{fill^+_{i+}}$ and $\mathds{1}_{fill^-_{i+}}$} to specify whether there are any ask and bid orders of the maker filled in the time interval \rev{$[t_{i+},t_{i+1})$} respectively. If the maker sends new orders and cancellation instructions at time $t_{i}-$, then \rev{$\mathds{1}_{fill^+_{i+}}$ and $\mathds{1}_{fill^-_{i+}}$} specify whether the orders sent at time $t_i-$ are filled in \rev{$[t_{i+},t_{i+1})$}. Otherwise if the maker does not send any new orders or cancellation instructions at $t_i-$, i.e, \rev{$\delta^+,\delta^-=o$}, then the two indicator functions specify whether the outstanding orders (if any) are filled in \rev{$[t_{i+},t_{i+1})$}. If the maker sends cancellation instructions without new orders, i.e, \rev{$\delta^+=\infty,\delta^-=\infty$}, then there will be no orders of the maker in \rev{$[t_{i+},t_{i+1})$} and \rev{$\mathds{1}_{fill^+_{i+}}=\mathds{1}_{fill^-_{i+}}=0$}.

We now describe the dynamics of system states \rev{$(w,p,q,r^+,r^-)$} from $t_i-$ to $t_{i+1}-$ for $i=0,1,..,N-1$. To begin with, we define the price changes \rev{$\Delta p_i:=p(t_{i+})-p(t_i)=\sum\limits_{j=\mathcal{N}(t_i)+1}^{\mathcal{N}({t_{i+}})}{X_j}$ and $\Delta p_{i+}:=p(t_{i+1})-p(t_{i+})=\sum\limits_{j=\mathcal{N}(t_{i+})+1}^{\mathcal{N}({t_{i+1}})}{X_j}$}. Then we can readily obtain that
\rev{\begin{eqnarray}
w_{i+1}&=&w_{i}+(p_i+0.5+r^+_i)\mathds{1}_{fill^+_i}-(p_i-0.5-r^-_i)\mathds{1}_{fill^-_i}  +p_{fill^+_{i+}}\mathds{1}_{fill^+_{i+}}-p_{fill^-_{i+}}\mathds{1}_{fill^-_{i+}}, \nonumber\\
 \label{w i to i+1}\\
p_{i+1} &=& p_i+\Delta p_i+\Delta p_{i+}, \label{p i to i+1}\\
q_{i+1}&=&q_i-\mathds{1}_{fill^+_i}+\mathds{1}_{fill^-_i}-\mathds{1}_{fill^+_{i+}}+\mathds{1}_{fill^-_{i+}}, \label{q i to i+1} \\
r^\pm_{i+1}&=&(1-\mathds{1}_{out^\pm_{i+1}})\cdot\infty+ \mathds{1}_{out^\pm_{i+1}}r_{out^\pm_{i+1}}, \label{a i to i+1}
\end{eqnarray}}
where
\rev{
\begin{equation*}
p_{fill^\pm_{i+}}:=
	\left\{
	\begin{aligned}
	&p_i\pm\max\{0.5+\delta^+_i,-0.5\pm\Delta p_i\}, &\delta^\pm_i \in \mathbb{\overline{Z}},\\
	&p_i\pm(0.5+r^\pm_i), &\delta^\pm_i= o,
	\end{aligned}
	\right.
\end{equation*}
\begin{minipage}{.5\linewidth}
\begin{equation*}
\mathds{1}_{out^\pm_{i+1}}:=
\left\{
\begin{aligned}
&1-\mathds{1}_{fill^\pm_{i+}}, &\delta^\pm\in\mathbb{\overline{Z}},\\
&1-\mathds{1}_{fill^\pm_{i}}-\mathds{1}_{fill^\pm_{i+}}, &\delta^\pm=o,
\end{aligned}
\right.
\end{equation*}
\end{minipage}
\begin{minipage}{.5\linewidth}
\begin{equation*}
r_{out^\pm_{i+1}}:=
\left\{
\begin{aligned}
&\delta^\pm_i \mp (\Delta p_i+\Delta p_{i+}),
&\delta^\pm\in\mathbb{\overline{Z}},\\
&r^\pm_i \mp (\Delta p_i+\Delta p_{i+}), &\delta^\pm=o,
\end{aligned}
\right.
\end{equation*}
\end{minipage}}

We briefly explain the dynamics of wealth and outstanding orders as others are straightforward to see. In the wealth dynamics, the market maker earns an amount equal to the execution price if an ask order is filled, and pays an amount equal to the execution price if a bid order is filled.
For the $i$-th period,  \rev{$p_{fill^+_{i+}}$} is the price of the ask order of the maker that is filled in the time interval \rev{$[t_{i+},t_{i+1})$} (if any). If \rev{$\delta^+\in \mathbb{\overline{Z}}$}, then the ask outstanding order will be canceled at time \rev{$t_{i+}$}. Thus, the ask order filled in \rev{$[t_{i+},t_{i+1})$} is the new order sent in the $i$-th action. Its execution price is the larger of its own price \rev{$p_i+0.5+\delta^+_i$} and the market best bid price at time \rev{$t_{i+}$}, i.e., \rev{$p_i-0.5+\Delta p_i$}. Otherwise, if \rev{$\delta^+=o$}, then the outstanding order will not be canceled in the $i$-th period and the execution price in \rev{$[t_{i+},t_{i+1})$} is the price of the ask outstanding order. It is similar for the bid side.

For the outstanding orders, taking the ask side as an example, \rev{$\mathds{1}_{out^+_{i+1}}$} specifies whether there exists an ask outstanding order at time $t_{i+1}-$ (1 if exists; 0 if not) and \rev{$r_{out^+_{i+1}}$} is its relative price comparing with the best \rev{ask} price at time $t_{i+1}-$. If \rev{$\delta^+\in \mathbb{\overline{Z}}$}, then such outstanding order exists if and only if the new ask order sent at time $t_i-$ is not filled in \rev{$[t_{i+},t_{i+1})$}. Otherwise if \rev{$\delta^+=o$}, then the order exists if and only the outstanding order for the $i$-the period (exists at time $t_i-$) is not filled in $[t_i,t_{i+1})$. If \rev{$\delta^+\in \mathbb{\overline{Z}}$}, the relative price is based on the price of the new order sent at time $t_i-$; otherwise it is based on the price of the outstanding order for the $i$-th period.

We next describe the dynamics from $t_N-$ to \rev{$t_{N+}-$}. We have
\rev{\begin{eqnarray}
w_{N+} &=&w_N+(p_N+0.5+r^+_N)\mathds{1}_{fill^+_N}-(p_N-0.5-r^-_N)\mathds{1}_{fill^-_N}, \label{w N to N.5} \\
p_{N+}&=&p_N+\Delta p_N, \label{p N to N.5}\\
q_{N+}&=&q-\mathds{1}_{fill^+_N}+\mathds{1}_{fill^-_N}, \label{q N to N.5} \\
r^\pm_{N+}&=&\mathds{1}_{fill^\pm_N}\cdot\infty+(1-\mathds{1}_{fill^\pm_N})(r^\pm_N\mp\Delta p_N). \label{a N to N.5}
\end{eqnarray}}
The main difference compared with the dynamics from $t_i$ to $t_{i+1}, i=0,1,..,N-1$ is that the market maker only unwinds his inventory position at time $t_N$ without posting new quotes. To see \eqref{a N to N.5}, note that if the ask outstanding order for the $N$-th period is filled, then there will be no ask outstanding orders at time $t_N-$. Otherwise, there will be an ask outstanding order at \rev{the relative price $p_N+0.5+r^+_N-(p_N+0.5+\Delta p_N)=r^+_N-\Delta p_N$ comparing with the best ask price at time $t_{N+}-$}. It is similar for the bid side.

\renewcommand{\theequation}{B\arabic{equation}}
\setcounter{equation}{0}
\section{Proofs of main results in Section~\ref{sec theoretical}}\label{proofs}
\subsection{Proof of Theorem \ref{value of orders tau}} \label{sec:proofOrderValue}
\begin{proof}
We only prove the results for the ask side. The bid side is similar.

\textbf{[Proof of Part (a)].}
For any \rev{$\delta^+<0$}, the ask order sent is immediately executed so
\rev{$\mathds{1}_{fill^+_{0,t_2',\delta^+}}\equiv1$}. It then follows from \eqref{H(xy)}
	that for any \rev{$\delta^+< 0$}, we have
	\rev{$H^{+}(0,t_2',\delta^+)=E[-0.5-\Delta p[0,t_2']]=-0.5$} by the martingale property of $p(\cdot)$.
	
For \rev{$\delta^+\ge0$}, from \eqref{H(xy)} we deduce that in order to establish  \eqref{eq:value1}, it suffices to show
	 \rev{\begin{equation}\label{conditionalvalue1}
	 E[\delta^++1-\Delta p[0,t_2'] \mid \mathds{1}_{fill^+_{0,t_2',\delta^+}}=1]=\dfrac{\lambda^+}{\lambda^++\lambda/2}.
	 \end{equation}}
To this end, we construct an appropriate Markov chain.
	
Denote the jump times of the two independent Poisson processes $\mathcal{N}(\cdot)$ and  \rev{$\mathcal{N}^+(\cdot)$ by $\tau_1, \tau_2, ...$ and $\tau^+_1, \tau^+_2, ...$} respectively. Define
	\begin{alignat*}{2}
	UA(t):=
	\left\{
	\begin{aligned}
	&1, \quad  \text{if} \; \{n\in\mathbb{N}: \tau_{\mathcal{N}(t)}<\rev{\tau^+_n}\leq t\}\neq \emptyset,\\
	&0,\quad \text{otherwise},
	\end{aligned}
	\right.
	& \hskip 4em &
	&UAsk(t):=(p(t),UA(t)).
 	\end{alignat*}
Here, $UA(t)$
	specifies if there are any `uninformed' buy market orders (\rev{$\mathcal{N}^+(\cdot)$}) have arrived at the best ask since the most recent jump of market price $\tau_{\mathcal{N}(t)}$ at time t.
	One can readily verify that $UAsk(t)$ is a continuous-time Markov chain with state space $\mathbb{Z}\times\{0,1\}$ and the following transition rates for any $p\in\mathbb{Z}$:
	\begin{alignat*}{2}
	(p,0) \longrightarrow
	\left\{
	\begin{aligned}
	&(p+1,0), \quad&\text{with rate} \ \lambda/2,\\
	&(p-1,0), \quad&\text{with rate}\  \lambda/2,\\
	&(p,1), \quad&\text{with rate} \  \rev{\lambda^+},\\
	\end{aligned}
	\right.
  & \hskip 5em &
	(p,1) \longrightarrow
	\left\{
	\begin{aligned}
	&(p+1,0), \quad&\text{with rate} \ \lambda/2,\\
	&(p-1,0), \quad&\text{with rate}\  \lambda/2,\\
  &(p,1), \quad&\text{with rate} \  \rev{\lambda^+},\\
	\end{aligned}
	\right.
	\end{alignat*}
	Also define:
	\rev{\begin{align*}
	&\tau_{fill_1}:=\inf\{t\geq0: UAsk(t)=(p(0)+\delta^+,1)\},\\
	&\tau_{fill_2}:=\inf\{t\geq0: UAsk(t)=(p(0)+\delta^++1,0)\},\\
	&\tau_{fill}:=\min\{\tau_{fill_1},\tau_{fill_2},t_2'\}.
	\end{align*}}
	If $\tau_{fill_1}<\min\{\tau_{fill_2},t_2'\}$, then the ask order sent at time 0 with relative price \rev{$\delta^+$} and without latency will be filled by an `uninformed' order before time $t_2'$; on the other hand, if $\tau_{fill_2}<\min\{\tau_{fill_1},t_2'\}$, then the mid-price will cross the price of ask order before time $t_2'$. The ask order will be filled before time $t_2'$ if and only if $\tau_{fill}<t_2'$. Then, we can \gao{infer from Equation (\ref{conditionalvalue1}) that}
	\rev{\begin{equation}\label{conditionalvalue2}
	\begin{split}
	&E[\delta^++1-\Delta p[0,t_2']\mid \mathds{1}_{fill^+_{0,t_2',\delta^+}}=1]\\
	=&E[p(0)+\delta^++1-p(\tau_{fill})\mid\tau_{fill}<t_2']
	-E[p(t_2')-p(\tau_{fill})\mid\tau_{fill}<t_2'],\\
	\end{split}
	\end{equation}}
	where the second term is zero by applying optional sampling theorem to the martingale $p(\cdot)$.
In addition, we note that \rev{if $\tau_{fill}=\tau_{fill_1}$, then $p(\tau_{fill})=p(0)+\delta^+$; if $\tau_{fill}=\tau_{fill_2}$, then $p(\tau_{fill})=p(0)+\delta^++1$}. Therefore, we infer that
	\rev{\begin{equation}\label{conditionalvalue3}
	\begin{split}
	E[p(0)+\delta^++1 -p(\tau_{fill})\mid\tau_{fill}<t_2']
	=\mathbb{P}(\tau_{fill}=\tau_{fill_1}\mid \tau_{fill}<t_2').
	\end{split}
	\end{equation}}
So to prove \eqref{conditionalvalue1}, it remains to show
	\begin{equation} \label{eq:inter-step}
	\mathbb{P}(\tau_{fill}=\tau_{fill_1}\mid \tau_{fill}<t_2') = \rev{\frac{\lambda^+}{\lambda^++ \frac{\lambda}{2}}}.
	\end{equation}
\gao{To this end, we consider the embedded discrete-time Markov chain (DTMC) of $UAsk(t)$. The embedded DTMC, denoted as $\{UAsk_n: n \ge 0 \}$ has the property that $UAsk_n\in\{(p(0) + \delta^+,1),(p(0) + \delta^++1,0)\}$ only if $UAsk_{n-1}\in\{(p(0)+\delta^+,0),(p(0) + \delta^+,1)\}$. In addition, it is clear from the transition rates of the continuous-time Markov chain $UAsk(t)$ that
for each $n,$
	\begin{align*}
\frac{\mathbb{P}(UAsk_n=(p(0) + \delta^+,1)\mid UAsk_{n-1}=(p(0) + \delta^+,0))}{\mathbb{P}(UAsk_n\in\{(p(0) + \delta^+,1),(p(0) + \delta^++1,0)\}\mid UAsk_{n-1}=(p(0) + \delta^+,0))}
	=\frac{\lambda^+}{\lambda^++\lambda/2}.
	\end{align*}
Together with the Markov property and the independence of the holding times and the jump transitions of the continuous time Markov chain, one can readily verify that \eqref{eq:inter-step} holds.}

 	\textbf{[Proof of Part (b)].}
 	By Equation \eqref{H(xy)}, we have, for any $t_1'\ge0$, $t_2'>0$ and \rev{$\delta^+\in \mathbb{\overline{Z}}$},
 	\gao{\begin{align*}
 	H^{+}(t_1',t_2',\delta^+)
 	=&E[(\max\{0.5+\delta^+-\Delta p[0,t_1'],-0.5\}-\Delta p[t_1',t_1'+t_2'])\mathds{1}_{fill^+_{t_1',t_2',\delta^+}}],
 	\end{align*}}
where the indicator function \rev{$\mathds{1}_{fill^+_{t_1',t_2',x}}$} specifies whether the ask order with relative price $x$ (which enters into the order book or executed at time $t_1'$) is filled before time $t_1'+t_2'$.
Note that \rev{$\{(p(t),\mathcal{N}^+(t)) : t\geq 0\}$} is a 2-dimensional process with stationary and independent increments. Then given $\Delta p[0,t_1']=k$ for any $k\in\mathbb{\overline{Z}}$, we can readily infer that
 \rev{\begin{align*}
 	\begin{aligned}
 	&E[(\max\{0.5+\delta^+-k,-0.5\}-\Delta p[t_1',t_1'+t_2'])\mathds{1}_{fill^+_{t_1',t_2',\delta^+}}
 	\mid \Delta p[0,t_1']=k]\\
 	=&E[(\max\{0.5+\delta^+-k,-0.5\}-\Delta p[0,t_2'])\mathds{1}_{fill^+_{0,t_2',\delta^+-k}}]\\
 	=&H^{+}(0,t_2',\delta^+-k).
 	\end{aligned}
 	\end{align*}}
 	Hence, we obtain the desired result from the tower property of the conditional expectations.

 	\textbf{[Proof of Part (c)].} First, suppose \rev{$\lambda^+\leq \lambda/2$}. When $t_1'=0$, by part (a) of Theorem \ref{value of orders tau}, it is clear that for any \rev{$\delta^+\in\mathbb{\overline{Z}}$}, we have $H^{+}(0,t_2',\delta^+)\le 0$.
 	Thus, for any $t_1',t_2'\ge0$, by part (b) of Theorem \ref{value of orders tau}, for any \rev{$\delta^+\in\mathbb{\overline{Z}}$},
 	\rev{\begin{align*}
 	H^{+}(t_1',t_2',\delta^+)=E[H^{+}(0, t_2',\delta^+-\Delta p[0,t_1'])]\leq 0.
 	\end{align*}}
 	
Next, suppose \rev{$\lambda^+>\lambda/2$}. When $t_1'=0$, by part (a) of Theorem \ref{value of orders tau}, when \rev{$\delta^+=0$},
 	\rev{\begin{align*}
 	H^{+}(0,t_2',0)=\left(\frac{\lambda^+}{\lambda^++\lambda/2}-0.5 \right) \cdot E\left[\mathds{1}_{fill^+_{0,t_2',0}} \right]>0.
 	\end{align*}}
 	Now consider $t_1'>0$.
 	By parts (a) and (b) of Theorem \ref{value of orders tau}, for any \rev{$\delta^+\in\mathbb{Z}$}, we have
 	\rev{\begin{align*}
 	H^{+}(t_1',t_2',\delta^+)
 	=&\sum\limits_{k=-\infty}^{\infty}H^{+}(0,t_2',\delta^+-k)\mathbb{P}(\Delta p[0,t_1']=k)\\
 	=&-0.5\mathbb{P}(\Delta p[0,t_1']\geq\delta^++1)+\left(\frac{\lambda^+}{\lambda^++\lambda/2}-0.5 \right)\sum\limits_{k=-\infty}^{\delta^+}
 	E[\mathds{1}_{fill^+_{0,t_2',\delta^+-k}}]
 	\mathbb{P}(\Delta p[0,t_1']=k)\\
 	\geq&-0.5\mathbb{P}(\Delta p[0,t_1']\geq\delta^++1)+
 	\left(\frac{\lambda^+}{\lambda^++\lambda/2}-0.5\right) \cdot
 	E[\mathds{1}_{fill^+_{0,t_2',0}}] \cdot \mathbb{P}(\Delta p[0,t_1']=\delta^+).
 	\end{align*}}
 	We claim that
 	\begin{align}\label{lemma eq:1}
 	\lim\limits_{\delta^+\rightarrow\infty}\frac{\mathbb{P}(\Delta p[0,t_1']\geq\rev{\delta^+})}{\mathbb{P}(\Delta p[0,t_1']=\rev{\delta^+-1})}=0,
 	\end{align}
 	and hence \rev{$H^{+}(t_1',t_2',\delta^+)>0$} if \rev{$\delta^+$} is sufficiently large.
 	
 	To prove Equation \eqref{lemma eq:1}, we first note that we only need to prove
 	\rev{\begin{align}\label{lemma eq:2}
 	\lim\limits_{\delta^+\rightarrow\infty}\frac{\mathbb{P}(\Delta p[0,t_1']=\delta^+)}{\mathbb{P}(\Delta p[0,t_1']=\delta^+-1)}=0.
 	\end{align}}
 	This is because, if Equation \eqref{lemma eq:2} holds, then there exists a constant $c\in(0,1)$, such that for \rev{$\delta^+$} sufficiently large, $\frac{\mathbb{P}(\Delta p[0,t_1']=\rev{\delta^+})}{\mathbb{P}(\Delta p[0,t_1']=\rev{\delta^+-1})}<c$.
 	Thus, as \rev{$\delta^+ \rightarrow \infty$},
 	\begin{align*}
 	\frac{\mathbb{P}(\Delta p[0,t_1']\geq\rev{\delta^+})}{\mathbb{P}(\Delta p[0,t_1']=\rev{\delta^+-1})}
 	\leq&\frac{(1+c+c^2+...)\mathbb{P}(\Delta p[0,t_1']=\rev{\delta^+})}{\mathbb{P}(\Delta p[0,t_1']=\rev{\delta^+-1})}
 	=\frac{\frac{1}{1-c}\mathbb{P}(\Delta p[0,t_1']=\rev{\delta^+})}{\mathbb{P}(\Delta p[0,t_1']=\rev{\delta^+-1})}\rightarrow 0.
 	\end{align*}
 	
 	Next, we prove Equation \eqref{lemma eq:2}. For any \rev{$0\leq\delta^+\in\mathbb{Z}$}, we have
 	\begin{equation}\label{lemma eq:3}
 	\begin{split}
 	&\frac{\mathbb{P}(\Delta p[0,t_1']=\rev{\delta^+})}{\mathbb{P}(\Delta p[0,t_1']=\rev{\delta^+-1)}}
 	=\frac{\sum\limits_{k=\rev{\delta^+}}^{\infty}\mathbb{P}(\Delta p[0,t_1']=\rev{\delta^+}\mid \mathcal{N}(t_1')=k)\mathbb{P}( \mathcal{N}(t_1')=k)}
 	{\sum\limits_{k=\rev{\delta^+}}^{\infty}\mathbb{P}(\Delta p[0,t_1']=\rev{\delta^+-1}\mid \mathcal{N}(t_1')=k-1)\mathbb{P}( \mathcal{N}(t_1')=k-1)},
 	\end{split}
 	\end{equation}
 	noting that if $k<\rev{\delta^+}$, then $\mathbb{P}(\Delta p[0,t_1']=\rev{\delta^+}\mid \mathcal{N}(t_1')=k)=0$.
 	For any \rev{$0 \leq\delta^+\leq k$}, we have
 	\rev{\begin{equation*}
 	\mathbb{P}(\Delta p[0,t_1']=\delta^+\mid \mathcal{N}(t_1')=k)
 	=\left\{
 	\begin{split}
 	&\binom{k}{\frac{k+\delta^+}{2}}\frac{1}{2^k}, &\text{if } k \text{ and } \delta^+ \text{ have the same parity},\\
 	&0, &\text{if } k \text{ and } \delta^+ \text{ have opposite parity}.
 	\end{split}
 	\right.
 	\end{equation*}}
 	Thus, for any $0\leq\rev{\delta^+}\leq k$, if $\rev{\delta^+}$ and $k$ have the same parity, then
 	\begin{align*}
 	&\frac{\mathbb{P}(\Delta p[0,t_1']=\rev{\delta^+}\mid \mathcal{N}(t_1')=k)\mathbb{P}( \mathcal{N}(t_1')=k)}
 	{\mathbb{P}(\Delta p[0,t_1']=\rev{\delta^+-1}\mid \mathcal{N}(t_1')=k-1)\mathbb{P}( \mathcal{N}(t_1')=k-1)}\\
 	=&\frac{\binom{k}{\frac{k+\rev{\delta^+}}{2}}\frac{1}{2^k}e^{-\lambda t_1'}(\lambda t_1')^k/k!}
 	{\binom{k-1}{\frac{k-1+\rev{\delta^+-1}}{2}}\frac{1}{2^{k-1}}e^{-\lambda t_1'}(\lambda t_1')^{k-1}/(k-1)!}
 	=\frac{\lambda t_1'}{k+\rev{\delta^+}}\leq \frac{\lambda t_1'}{2\rev{\delta^+}}.
 	\end{align*}
 	Therefore, by Equation \eqref{lemma eq:3}, we obtain that
 	$\frac{\mathbb{P}(\Delta p[0,t_1']=\rev{\delta^+})}{\mathbb{P}(\Delta p[0,t_1']=\rev{\delta^+-1})}
 	\leq \frac{\lambda t_1'}{2\rev{\delta^+}}\rightarrow 0,$
 	as $\rev{\delta^+}\rightarrow \infty$.
 	Thus, Equation \eqref{lemma eq:2} holds and the proof of part (c) is complete for the ask side. The proof for the bid part is similar and hence omitted.
\end{proof}

\subsection{Proof of Theorem \ref{BellmanThm}}
\begin{proof}
	We prove Theorem \ref{BellmanThm} by backward induction and using the Bellman Equation~\eqref{Bellman}. Recall the definitions of \rev{$\mathds{1}_{fill^+_i}, \mathds{1}_{fill^+_{i+}}, \mathds{1}_{fill^-_i},  \mathds{1}_{fill^-_{i+}}, \Delta p_i, \Delta p_{i+}, p_{fill^+_{i+}}, p_{fill^-_{i+}}$} given in Section~\ref{sec:dynamics}.

 For $i=N$, by the Bellman equation~\eqref{Bellman}, for any \rev{$s=(w,p,q,r^+,r^-)\in S$}, we have
	\rev{\begin{align*}
	v_N(s)
	=&E[w+(p+0.5+r^+)\mathds{1}_{fill^+_N}-(p-0.5-r^-)\mathds{1}_{fill^-_N}+(p+\Delta p_N)(q-\mathds{1}_{fill^+_N}+\mathds{1}_{fill^-_N})\\
	\quad&-0.5|q-\mathds{1}_{fill^+_N}+\mathds{1}_{fill^-_N}|\;\big| s_N=s] \\
	=&w+pq+H^{+}(0,\Delta\tau,r^+)+H^{-}(0,\Delta\tau,r^-)-0.5E[\;|q-\mathds{1}_{fill^+_0}+\mathds{1}_{fill^-_0}|\;\big | (r^+_0,r^-_0)=(r^+,r^-)],\\
	=&w+pq+g_N(q,r^+,r^-),
	\end{align*}}
\gao{where we have used the stationarity of the MDP and the definitions of $H^{\pm}$ and $g_N$}.
	
For $i=0,1,...,N-1$, assume \rev{$v_{i+1}(s)=w+pq+g_{i+1}(q,r^+,r^-)$} for any $s=(w,p,q,r^+,r^-)\in S$, then for any \rev{$s=(w,p,q,r^+,r^-)\in S$}, we can compute that
	\rev{\begin{align*}
	v_{i}(s)
	=&\max\limits_{(\delta^+,\delta^-)\in A_s}E[v_{i+1}(s_{1})\mid s_0=s, (\delta^+_0,\delta^-_0)=(\delta^+,\delta^-)]\\
	=&\max\limits_{(\delta^+,\delta^-)\in A_s}E
	[w+(p+0.5+r^+)\mathds{1}_{fill^+_0}-(p-0.5-r^-)\mathds{1}_{fill^-_0}+p_{fill^+_{0+}}\mathds{1}_{fill^+_{0+}}-p_{fill^-_{0+}}\mathds{1}_{fill^-_{0+}}\\
	\quad  &+(p+\Delta p_0+\Delta p_{0+})(q-\mathds{1}_{fill^+_0}+\mathds{1}_{fill^-_0}-\mathds{1}_{fill^+_{0+}}+\mathds{1}_{fill^-_{0+}})\\
	\quad  &+g_{i+1}(q_1,r^+_1,r^-_1) \mid s_0=s,(\delta^+_0,\delta^-_0)=(\delta^+,\delta^-)],
	\end{align*}}
	\gao{where we have used the Bellman equation~\eqref{Bellman}, the assumption for $v_{i+1}$ and the system dynamics in Equations \eqref{w i to i+1}-\eqref{q i to i+1}}.
\rev{Reorganizing the terms and using the fact that $\Delta p_{0+}$ is independent with $\mathds{1}_{fill^+_0}$ and $\mathds{1}_{fill^-_0}$, we obtain}
	\rev{\begin{align*}
	v_{i}(s)
	=&w+pq+H^{+}(0,\Delta\tau,r^+)+H^{-}(0,\Delta\tau,r^-)\\
	&+\max\limits_{(\delta^+,\delta^-)\in A_s}
	\Big\{{E[(p_{fill^+_{0+}}-p-\Delta p_0-\Delta p_{0+})\mathds{1}_{fill^+_{0+}}]}\\
	&+{E[(p+\Delta p_0+\Delta p_{0+}-p_{fill^-_{0+}})\mathds{1}_{fill^-_{0+}}]}\\
	\quad&+E[g_{i+1}(q_1,r^+_1,r^-_1)\mid (q_0,r^+_0,r^-_0)=(q,r^+,r^-), (\delta^+_0,\delta^-_0)=(\delta^+,\delta^-)]\Big \}.
	\end{align*}}
	By the definition of the function $G_i$, it remains to prove that, for any possible \rev{$(r^+,r^-,\delta^+,\delta^-)$},
	\rev{\begin{equation} \label{value decom}
	\begin{aligned}
	&H^{+}(0,\Delta\tau,r^+)+E[(p_{fill^+_{0+}}-p-\Delta p_0-\Delta p_{0+})\mathds{1}_{fill^+_{0+}}\mid (r^+_0,\delta^+_0)=(r^+,\delta^+)]=H^{+}_{act}(r^+,\delta^+),  \\
	&H^{-}(0,\Delta\tau,r^-)+E[(p+\Delta p_0+\Delta p_{0+}-p_{fill^-_{0+}})\mathds{1}_{fill^-_{0+}}\mid (r^-_0,\delta^-_0)=(r^-,\delta^-)]=H^{-}_{act}(r^-,\delta^-).
		\end{aligned}
	\end{equation}}
We prove the first equation (i.e. ask side) in \eqref{value decom}. For \rev{$\delta^+\in\mathbb{\overline{Z}}$}, it directly follows from the definitions of \rev{$p_{fill^+_{0+}}$} (given after equation \eqref{a i to i+1}) and \rev{$H^{+}_{act}(r^+,\delta^+)$}.
	For \rev{$\delta^+=o$}, we need to show
	\rev{\begin{equation}\label{value decom2}
	H^{+}(0,\Delta\tau,r^+)+E[(r^+-\Delta p_0-\Delta p_{0+})\mathds{1}_{fill^+_{0+}}\mid (r^+_0,\delta^+_0)=(r^+,o)]=H^{+}(0,\Delta t,r^+)
	\end{equation}}
This can also be readily verified by using the definition of \rev{$H^{+}$} in \eqref{H(xy)} and the fact that the price process $p(\cdot)$ is a martingale. The proof is thus complete.
\end{proof}


\subsection{Proof of Theorem \ref{net profit}}\label{proof net profit}  

To prove Theorem \ref{net profit}, we need a lemma showing that when other parameters fixed, the \rev{expected profit} of the market maker \rev{$\mathcal{P}$} is a non-decreasing function of the number of quoting periods $N$.
\begin{lemma}\label{NP vs N}
	Fix latency $\Delta \tau \ge 0$ and other
	parameters $\lambda,\rev{\lambda^+,\lambda^-}, \Delta t, \overline{q},\underline{q}$. We have \rev{$\mathcal{P}$} is a non-decreasing function of $N$.
\end{lemma}
\begin{proof}[Proof of Lemma \ref{NP vs N}]
The main idea is given as follows. Comparing two MDP problems with $N$ and $N+1$ periods respectively, the value function at time $t_1$ in the latter is the same as the value function at $t_0$ (i.e., the initial one) in the former, because they can be computed by the same backward recursion (Bellman equation) from the same terminal value function.
For \rev{$s=(w,p,0,\infty,\infty)$}, the value function at $t_0$ in the $N+1$ period problem is greater than or equal to that at $t_1$ in the same problem, because the maker can choose to post no orders in the initial action. Thus, \rev{$\mathcal{P}$} with $N+1$ periods is greater than or equal to that with $N$ periods.

Mathematically, by Theorem \ref{BellmanThm} and equation \eqref{NP 1}, \rev{$\mathcal{P}=g_0(0,\infty,\infty)$} is a function of $N$. Denote this function by \rev{$f_{\mathcal{P}}(N)$}. Clearly we have \rev{$f_{\mathcal{P}}(0)=0$}.
For the two MDP problems with $N=n\geq 1$ and $N=n+1$, denote the value functions and corresponding $g$ function by $v^{n}_i(s), g^{n}_i(s), i=0,1,...,n$ and $v^{n+1}_i(s), g^{n+1}_i(s), i=0,1,...,n+1$ respectively.
By Theorem \ref{BellmanThm}, functions $v^{n}_0(s)$ and $v^{n+1}_1(s)$ are computed by the same backward recursion process starting from the same function
\rev{\begin{align*}
v^{n}_n(s)=v^{n+1}_{n+1}(s)=&w+pq+H^{+}(0,\Delta \tau,r^+)+H^{-}(0,\Delta \tau,r^-)\\
&-0.5E[\;|q-\mathds{1}_{fill^+_0}+\mathds{1}_{fill^-_0}|\; \big| (r^+_0,r^-_0)=(r^+,r^-)],
\end{align*}}
for any \rev{$s=(w,p,q,r^+,r^-)\in S$}. Therefore, for any $s\in S$, $v^{n}_0(s)=v^{n+1}_1(s)$. By Theorem \ref{BellmanThm}, we have for any $w,p\in \mathbb{Z}$
\rev{\begin{align*}
v^{n+1}_0(w,p,0,\infty,\infty) =&w+p\cdot 0+	\max\limits_{(\delta^+,\delta^-)\in A_s} \Big\{H^{+}_{act}(\infty,\delta^+)+H^{-}_{act}(\infty,\delta^-)+E[g^{n+1}_{1}(q_1,r^+_1,r^-_1)\\
& \qquad \qquad  \qquad \qquad \qquad \qquad \mid (q_0,r^+_0,r^-_0)=(0,\infty,\infty),(\delta^+_0,\delta^-_0)=(\delta^+,\delta^-)]\Big\}\\
\geq&w+H^{+}_{act}(\infty,\infty)+H^{-}_{act}(\infty,\infty)\\
&+E[g^{n+1}_{1}(q_1,r^+_1,r^-_1)
\mid (q_0,r^+_0,r^-_0)=(0,\infty,\infty),(\delta^+_0,\delta^-_0)=(\infty,\infty)]\\
=&w+g_{1}^{n+1}(0,\infty,\infty)=v^{n+1}_1(w,p,0,\infty,\infty).
\end{align*}}
Thus, we obtain that
\rev{\begin{align*}
f_{\mathcal{P}}(n+1)&=v^{n+1}_0(w,p,0,\infty,\infty)-w\\
&\geq v^{n+1}_1(w,p,0,\infty,\infty)-w=v^{n}_0(w,p,0,\infty,\infty)-w=f_{\mathcal{P}}(n).
\end{align*}}
The proof is therefore complete.
\end{proof}


\begin{proof} [Proof of Theorem \ref{net profit}.]
We first prove part (1). By part (c) of Theorem~\ref{value of orders tau}, if \rev{$\lambda^+\leq \lambda/2$ and $\lambda^-\leq \lambda/2$}, functions \rev{$H^{+},H^{-}$} are always non-positive and hence by the definitions of \rev{$H^{+}_{act},H^{-}_{act}$}, for any possible \rev{$(r^+,r^-,\delta^+,\delta^-)$, $H^{+}_{act}(r^+,\delta^+), H^{-}_{act}(r^-,\delta^-)$} are both non-positive.
We prove \rev{$g_i(q,r^+,r^-)\leq 0$}, for $i=0,1,...,N$ and any admissible \rev{$(q,r^+,r^-)$} by the backward recursion in Theorem \ref{BellmanThm}. For $i=N$, it holds directly from Equation \eqref{Bellman g_i}.
Suppose for some $i$ ($1\leq i\leq N$), \rev{$g_i(q,r^+,r^-)\leq 0$} for any admissible \rev{$(q,r^+,r^-)$}. Then from equations (\ref{Bellman g_i}) and (\ref{H_i 5}), we obtain that 
\rev{\begin{align*}
g_{i-1}(q,r^+,r^-)
=&\max\limits_{(\delta^+,\delta^-)\in A_s} \Big\{H^{+}_{act}(r^+,\delta^+)+H^{-}_{act}(r^-,\delta^-)+E[g_i(q_1,r^+_1,r^-_1)\\
 & \quad\ \quad \mid (q_0,r^+_0,r^-_0)=(q,r^+,r^-),(\delta^+_0,\delta^-_0)=(\delta^+,\delta^-)]\Big \}\leq 0.
\end{align*}}
Therefore, \rev{$g_i(q,r^+,r^-)\leq 0$}, for $i=0,1,...,N$ and any admissible \rev{$(q,r^+,r^-)$}. It follows from Equation \eqref{NP 1} that
$\mathcal{P}=g_0(0,\infty,\infty)\leq 0.$ Since $\mathcal{P}$ is lower bounded by zero, we get \rev{$\mathcal{P}=0$}.

Next, we prove part (2). Suppose \rev{$\lambda^+>\lambda/2$ and $\lambda^->\lambda/2$}. By part (c) of Theorem~\ref{value of orders tau}, there exists a quote pair, denoted by \rev{($\delta^{+,M},\delta^{-,M})\in\mathbb{\overline{Z}}^2$}, such that \rev{$H^{+}(\Delta\tau,\Delta t,\delta^{+,M})>0$ and
$H^{-}(\Delta\tau,\Delta t,\delta^{-,M})>0$}. Obviously, we have \rev{$\delta^{+,M},\delta^{-,M}\in \mathbb{Z}$} and if $\Delta\tau=0$ then \rev{$\delta^{+,M},\delta^{-,M}\ge 0$}, otherwise their order values will be 0 or $-0.5$.
The main idea is that we construct an admissible policy using this quote pair \rev{$(\delta^{+,M},\delta^{-,M})$}, following which the expected profit is positive if $N$ is a sufficiently large even number. Then, $N_{min}$ exists since \rev{$\mathcal{P}$} is a non-increasing function of $N$ by Lemma \ref{NP vs N}. We also give an upper bound of $N_{min}$. 

First, we define the admissible policy.
Denote by $v^{\pi}_i(s)$ the expected \rev{$\mathcal{W}$} following any admissible policy $\pi=\{f_i: i=0,1,...N\}$ starting from time $t_i$ with an initial state \rev{$s=(w,p,q,r^+,r^-)\in S$}.
We consider an admissible policy $\tilde{\pi}=\{\tilde{f_i}: i=0,1,...N\}$ in any of our MDP problem with an even $N\geq4$, i.e., $N=2K$ for some $2\leq K\in \mathbb{N}$.
Starting from time $0$ with an initial state \rev{$(w,p,0,\infty,\infty)$} for any $w,p\in\mathbb{Z}$, $\tilde{\pi}$ is defined as follows. For $i=1,3,5,...,N-1$, i.e., $i$ is odd, for any \rev{$w,p\in \mathbb{Z}, r^+,r^-\in \mathbb{\overline{Z}}$,  $q\in\{-1, 0, 1\}$ such that $(w,p,q,r^+,r^-)\in S$}, define
\rev{\begin{equation*}
\tilde{f}_i(w,p,q,r^+,r^-):=(\infty,\infty).
\end{equation*}}
For $i=2,4,6,...,N-2$, i.e., $i$ is even except 0 and $N$, for any $w,p\in \mathbb{Z}$,  $q\in\{-1, 0, 1\}$, define
\rev{\begin{equation*}
\tilde{f}_i(w,p,q,\infty,\infty):=\left\{
\begin{split}
&(\delta^{+,M},\infty), \quad &q= 1,\\
&(\infty,\infty), \quad &q=0,\\
&(\infty,\delta^{-,M}), \quad &q= -1.\\
\end{split}
\right.
\end{equation*}}
For any $w,p\in\mathbb{Z}$, define \rev{$\tilde{f}_0(w,p,0,\infty,\infty):=(\delta^{+,M},\delta^{-,M})$}. For a state $s$, if \rev{$(\delta^+,\delta^-)=\tilde{f}_i(s)$}, we also write \rev{$\delta^+=\tilde{f}^+_i(s), \delta^-=\tilde{f}^-_i(s)$}, $i=0,1,...,N-1$. Under this policy, at time $t_i, i=2,4,6,....,N-2$,
the maker will post no orders if his inventory is 0, sell one unit at the relative price \rev{$\delta^{+,M}$} if his inventory is 1, and buy one unit at \rev{$\delta^{-,M}$} if his inventory is -1. At time $t_i, i=1,3,5,....,N-1$, the maker cancels old orders if any and does not post any new orders. At time $t_N$, the maker unwinds his inventory if any. At time $t_0=0$, there are no inventories nor outstanding orders and the maker quotes at \rev{$(\delta^{+,M},\delta^{-,M})$}. Therefore, at time $t_i$, $i=0,2,4,6,...,N-2$, there are no outstanding orders. Moreover, the inventory of the maker always belongs to  $\{-1, 0, 1\}$.

Then, we give the backward recursion for this policy. Standard arguments in MDP theory yield that for $i=0,2,4,...,N-2$, for any $w,p\in \mathbb{Z}$,  $q\in\{-1, 0, 1\}$,
\rev{\begin{align*}
v^{\tilde{\pi}}_i(w,p,q,\infty,\infty)
=&E[v^{\tilde{\pi}}_{i+2}(w_{i+2},p_{i+2},q_{i+2},\infty,\infty) \mid (w_i,p_i,q_i,r^+_i,r^-_i)=(w,p,q,\infty,\infty),\\
\quad & (\delta^+_i,\delta^-_i)=\tilde{f}_i(w,p,q,\infty,\infty), (\delta^+_{i+1},\delta^-_{i+1})=(\infty,\infty)].
\end{align*}}
Using a similar argument as in the proof of Theorem \ref{BellmanThm}, we obtain that $ v_i^{\tilde{\pi}}(w,p,q,\infty,\infty)
=w+pq+g_i^{\tilde{\pi}}(q),$ for $i=2,4,6,...,N$ and for any $w,p\in \mathbb{Z}$, $q\in\{-1,0,1\}$.
Here $g^{\tilde{\pi}}_N(q)=-0.5|q|$, and
\rev{\begin{equation}\label{thm4 gi}
\begin{aligned}
g^{\tilde{\pi}}_i(q)
=&H^{+}(\Delta\tau,\Delta t,\tilde{f}^+_i(w,p,q,\infty,\infty))+H^{-}(\Delta\tau,\Delta t,\tilde{f}^-_i(w,p,q,\infty,\infty))\\
&+E[g^{\tilde{\pi}}_{i+2}(q_2)\mid (q_0,r^+_0,r^-_0)=(q,\infty,\infty),\\
\quad &  \qquad \qquad (\delta^+_0,\delta^-_0)=\tilde{f}_i(w,p,q,\infty,\infty),
(\delta^+_{1},\delta^-_{1})=(\infty,\infty)],\text{ for } i=2,4,6,...,N-2.
\end{aligned}
\end{equation}}
To understand Equation \ref{thm4 gi}, note that for each order sent at time $t_i, i=2,4,6,...N-2$, the lifetime is $\Delta t$ due to the cancellation instruction in the next period and there are no outstanding orders at time $t_i, i=2,4,6,...N-2$.
Similarly, for any $w,p\in \mathbb{Z}$, we have
\rev{$v_0^{\tilde{\pi}}(w,p,0,\infty,\infty)
=w+p\cdot0+g^{\tilde{\pi}}_0(0)
=w+g_0^{\tilde{\pi}}(0),$}
where
\rev{\begin{equation}\label{thm4 g0}
\begin{aligned}
g^{\tilde{\pi}}_0(0)
=&H^{+}(\Delta\tau,\Delta t,\delta^{+,M})+H^{-}(\Delta\tau,\Delta t,\delta^{-,M})
+E[g^{\tilde{\pi}}_{2}(q_2)\mid (q_0,r^+_0,r^-_0)=(0,\infty,\infty),\\
& \quad \qquad (\delta^+_0,\delta^-_0)=(\delta^{+,M},\delta^{-,M}),
(\delta^+_{1},\delta^-_{1})=(\infty,\infty)].
\end{aligned}
\end{equation}}

Next, we prove that if $N$ is sufficiently large, then $g^{\tilde{\pi}}_0(0)>0$. To show this, by Equation \eqref{thm4 g0}, we only need to prove  $g^{\tilde{\pi}}_{2}(q)\geq0$ for $q=-1,0,1$, since \rev{$H^{+}(\Delta\tau,\Delta t,\delta^{+,M})+H^{-}(\Delta\tau,\Delta t,\delta^{-,M})>0$}.
First, we prove $g^{\tilde{\pi}}_{2}(0)=0$.
By Equation \eqref{thm4 gi}, we obtain that, for $i=2,4,6,...,N-2$,
\rev{\begin{align*}
g^{\tilde{\pi}}_i(0)
=&H^{+}(\Delta\tau,\Delta t,\infty)+H^{-}(\Delta\tau,\Delta t,\infty)\\
& \quad +E[g^{\tilde{\pi}}_{i+2}(q_2)\mid (q_0,r^+_0,r^-_0)=(0,\infty,\infty),(\delta^+_0,\delta^-_0)=(\infty,\infty),
(\delta^+_{1},\delta^-_{1})=(\infty,\infty)]\\
=&g^{\tilde{\pi}}_{i+2}(0).
\end{align*}}
Hence, for $i=2,4,6,...,N-2$, $g^{\tilde{\pi}}_i(0)=g^{\tilde{\pi}}_N(0)=0$.
Then, we prove that if $N$ is sufficiently large, then $g^{\tilde{\pi}}_{2}(\pm1)>0$. Recall that \rev{$\mathds{1}_{fill^+_{\Delta\tau,\Delta t,\delta^{+,M}}}$} specifies whether the ask order sent at time 0 and the relative price \rev{$\delta^{+,M}$} with latency $\Delta\tau$ is filled in the time interval $[\Delta\tau,\Delta\tau+\Delta t)$.
It can be readily verified that given \rev{$\delta^+_0=\delta^{+,M}$} we have \rev{$\mathds{1}_{fill^+_{\Delta\tau,\Delta t,\delta^{+,M}}}=\mathds{1}_{fill^+_{0+}}+\mathds{1}_{fill^+_1}$}.
Intuitively, the ask order is filled if and only if either it is filled in the time interval $[\Delta\tau,\Delta t)$, which is represented by \rev{$\mathds{1}_{fill^+_{0+}}=1$}, or it stays as an outstanding order at time $t_1=\Delta t$, and then filled in the time interval $[\Delta t, \Delta \tau+\Delta t)$, which is represented by \rev{$\mathds{1}_{fill^+_{1}}=1$}.
Denote the fill probability of this ask order quoted at relative price \rev{$\delta^{+,M}$ by $p^+$}, where
\rev{\begin{equation}\label{eq:pa}
p^+=\mathbb{P}(\mathds{1}_{fill^+_{\Delta \tau,\Delta t,\delta^{+,M}}}=1).
\end{equation}}
By Equation \eqref{thm4 gi}, for $i=2,4,6,...,N-2$, we have
\rev{\begin{align*}
g^{\tilde{\pi}}_i(1)
=&H^{+}(\Delta \tau,\Delta t,\delta^{+,M})
+E[g^{\tilde{\pi}}_{i+2}(1-\mathds{1}_{fill^+_{0+}}-\mathds{1}_{fill^+_1})
\mid (q_0,r^+_0,r^-_0)=(1,\infty,\infty),\\
\quad&(\delta^+_0,\delta^-_0)=(\delta^{+,M},\infty),
(\delta^+_{1},\delta^-_{1})=(\infty,\infty)]\\
=&H^{+}(\Delta\tau,\Delta t,\delta^{+,M})+(1-p^+)g^{\tilde{\pi}}_{i+2}(1),
\end{align*}}
where we have used $g^{\tilde{\pi}}_{i+2}(0)=0$ and the fact that the random variable \rev{$\mathds{1}_{fill^+_{\Delta\tau,\Delta t,\delta^{+,M}}}$} does not depend on \rev{$q_0,r^+_0,r^-_0,\delta^+_0,\delta^-_0,\delta^+_1$ or $\delta^-_1$}.
It is similar for the bid side. Define the fill probability of the bid order quoted at \rev{$\delta^{-,M}$} by
\rev{$p^-$}, where
\rev{\begin{equation} \label{eq:pb}
p^-=\mathbb{P}(\mathds{1}_{fill^-_{\Delta\tau,\Delta t,\delta^{-,M}}}=1).
\end{equation}}
 For $i=2,4,6,...,N-2$, we have
\rev{\begin{align*}
g^{\tilde{\pi}}_i(-1)
=&H^{-}(\Delta\tau,\Delta t,\delta^{-,M})+(1-p^-)g^{\tilde{\pi}}_{i+2}(-1).
\end{align*}}
Recall $N=2K$. Solving the above recursive equations for $g^{\tilde{\pi}}_i(\pm1)$ with $g^{\tilde{\pi}}_N(1)=g^{\tilde{\pi}}_N(-1)=-0.5$, we obtain that
\rev{\begin{eqnarray}
g^{\tilde{\pi}}_2(\pm 1)&=&\left(-0.5-\frac{H^{\pm}(\Delta\tau,\Delta t,\delta^{\pm,M})}{p^\pm} \right)(1-p^\pm)^{K-1}+\frac{H^{\pm}(\Delta\tau,\Delta t,\delta^{\pm,M})}{p^\pm}.
\end{eqnarray}}
Note that \rev{$0<p^\pm<1$ because $\delta^{\pm,M}\in \mathbb{Z}$ and if $\Delta \tau=0$, we have $\delta^{\pm,M}\ge0$ }. In addition, \rev{$H^{\pm}(\Delta\tau,\Delta t,\delta^{\pm,M})>0$}.
Thus, if $N$ is sufficiently large, then $g^{\tilde{\pi}}_2(\pm1)>0$ and hence $g^{\tilde{\pi}}_0(0)>0$.

Finally, since $\tilde{\pi}$ is an admissible policy, by Equation \eqref{NP 1}, we obtain that
\rev{\begin{equation*}
\mathcal{P}\geq v_0^{\tilde{\pi}}(w,p,0,\infty,\infty)-w=g_0^{\tilde{\pi}}(0)>0.
\end{equation*}}
Note that the backward recursion for value functions in each period depends on the model parameters $\lambda,\rev{\lambda^+,\lambda^-}, \Delta\tau, \Delta t, \overline{q}$ and $\underline{q}$.
 It follows from the monotonicity of \rev{$f_{\mathcal{P}}(N)$} (Lemma \ref{NP vs N}) and the fact \rev{$f_{\mathcal{P}}(N)>0$} when $N$ is even and sufficiently large that there exists a constant integer $N_{min}\geq1$ depending on $\lambda, \rev{\lambda^+,\lambda^-}, \Delta\tau, \Delta t, \overline{q}$ and $\underline{q}$, such that \rev{$f_{\mathcal{P}}(N)>0$} if and only if $N\geq N_{min}$. For an upper bound of $N_{min}$, define
\rev{\begin{equation*}
\overline{N}_{min}:=2\max\left\{\left\lceil \frac{\ln{\frac{H^{+}(\Delta\tau,\Delta t,\delta^{+,M})}{H^{+}(\Delta\tau,\Delta t,\delta^{+,M})+0.5p^+}}}{\ln{(1-p^+)}}\right\rceil,
\left\lceil\frac{\ln{\frac{H^{-}(\Delta\tau,\Delta t,,\delta^{-,M})}{H^{-}(\Delta\tau,\Delta t,\delta^{-,M})+0.5p^-}}}{\ln{(1-p^-)}}\right\rceil\right\}+2,
 \end{equation*}}
where \rev{$p^\pm$} are the fill probabilities of ask and bid orders given in \eqref{eq:pa} and \eqref{eq:pb}.
It can be readily verified that when $N\geq \overline{N}_{min}$, $g^{\tilde{\pi}}_2(\pm1)>0$ and hence $g^{\tilde{\pi}}_0(0)>0$. Therefore, $\overline{N}_{min}\geq N_{min}$. Now the proof for part (2) is complete.
\end{proof}

\section{Estimations of \rev{$\lambda^+$ and $\lambda^-$} for a hypothetical market maker} \label{sec:est-lambda-a}
In this section, we discuss how to estimate \rev{$\lambda^+$ and $\lambda^-$} for a particular market maker.
In our model, there are two cases for an order's execution. The first is that the order is filled and the mid-price does not move at the execution time, called by Type I event. The second is that the order is filled due to a price move, called by Type II event. Using similar arguments as in the proof of Part (a) of Theorem \ref{value of orders tau}, one can show that for the best bid order (\rev{$\delta^-=0$}) sent at time 0 with latency $\Delta\tau\ge0$, we have
\rev{\begin{equation} \label{estimationformula}
\mathbb{P}[p(t_{exe})=p(0) \mid p(\Delta\tau)=p(0), t_{exe}<t_{first move}]=\frac{\lambda^-}{\lambda^-+0.5\lambda},
\end{equation}}
where $t_{exe}$ is the execution time of the bid order and
\begin{equation*}
t_{firstmove}:=\inf\{t\ge0 \mid p(t)=p(0)+1 \}
\end{equation*}
is the first time that the market price moves up by 1 tick.
In \eqref{estimationformula}, the left hand side is the probability of a Type I event for the bid order conditioned on that the bid order enters into the best bid level correctly and the bid order is filled before market price moves up. Based on \eqref{estimationformula} and actual order submissions and executions, a market maker can first estimate the conditional probability and then use the estimated value of $\lambda$ to compute the estimated value of \rev{$\lambda^-$}.

We illustrate this estimation procedure using artificial order simulations based on NASDAQ's TotalView-ITCH message data. The simulation procedure is similar as in \cite{Moallemi2}. Without loss of generality we focus on bid orders for an illustration.
We randomly insert 500 artificial orders (subject to latency) at best bid price level in the orderbook each day, and
update the status of the orders at every event time under the matching rule of the exchange.
The artificial orders are assumed to be of infinitesimal size, so they do not impact the market. Given a fixed latency level,
only those orders for which the market price does not move in the latency window are considered. For any of such bid bids, suppose the order is filled before the market best bid price moves up. If at the execution time the market price does not move down, then the order's execution is counted as a Type I event; otherwise it is counted as a Type II event. Denote the total number of Type I events by $n_1$ and that of Type II events by $n_2$. By \eqref{estimationformula}, we have \rev{$\frac{n_1}{n_1+n_2}\approx \frac{\lambda^-}{\lambda^-+0.5\lambda}$} and thus the maker estimates \rev{$\lambda^-$} by $\frac{n_1\lambda}{2n_2}$. The rate \rev{$\lambda^+$} can be estimated similarly.

We now report the estimation results based on the data of GE in the forth quarter of 2016.
As for a given $\lambda$, the estimated values of \rev{$\lambda^+$ and $\lambda^-$} essentially depend on the ratio between the number of Type I events and that of Type II events, we focus on reporting the ratio \rev{$\lambda^-/(0.5 \lambda)$}.
As the results vary among different days, we show the average ratios for the quarter for different latency levels. See Table~\ref{esimationlambda^a}. We can observe that the ratio decreases from $1.07>1$, i.e., \rev{$\lambda^->0.5\lambda$} to $0.94<1$, i.e., \rev{$\lambda^-<0.5\lambda$} as latency increases. The reason is that given the market environment, a hypothetical market maker with lower latency can obtain better queue positions for his orders and suffer less from adverse selections on average. 
\begin{table}[h]
	\begin{center}
			\caption{\rev{$\lambda^-/(0.5 \lambda)$} for different latency for GE in 2016Q4.}
			{\begin{tabular}{@{}lccccccc}
					\hline
					$\Delta\tau (ms)$ &0 & 10  & 50 & 100 &  500 & 1000 & 2000 \\
					ratio \rev{$\lambda^-/(0.5 \lambda)$} & 1.07  & 1.04 & 1.03& 1.02  &  1.00 & 0.97& 0.94 \\
					\hline
			\end{tabular}}
			\label{esimationlambda^a}
	\end{center}
\end{table}


\renewcommand{\theequation}{C\arabic{equation}}
\setcounter{equation}{0}

\end{document}